\documentclass[pra,aps,twocolumn,twoside,superscriptaddress]{revtex4}

\usepackage{amsmath,amsfonts,amssymb,caption,color,epsfig,graphics,graphicx,hyperref,latexsym,mathrsfs,revsymb,theorem,url,verbatim,enumerate,epstopdf,tikz,float,multirow,booktabs,appendix}
\usepackage{ragged2e}
\hypersetup{colorlinks,linkcolor={blue},citecolor={red},urlcolor={blue}}
\usetikzlibrary{arrows, decorations.markings}

\tikzstyle{vecArrow} = [thick, decoration={markings,mark=at position
	1 with {\arrow[semithick]{open triangle 60}}},
double distance=1.4pt, shorten >= 5.5pt,
preaction = {decorate},
postaction = {draw,line width=1.4pt, white,shorten >= 4.5pt}]
\tikzstyle{innerWhite} = [semithick, white,line width=1.4pt, shorten >= 4.5pt]

\newtheorem{definition}{Definition}
\newtheorem{proposition}[definition]{Proposition}
\newtheorem{lemma}[definition]{Lemma}

\newtheorem{theorem}[definition]{Theorem}
\newtheorem{corollary}[definition]{Corollary}
\newtheorem{conjecture}[definition]{Conjecture}

\newtheorem{remark}[definition]{Remark}
\newtheorem{example}[definition]{Example}
\newtheorem{question}[definition]{Question}

\def\bcj{\begin{conjecture}}
	\def\ecj{\end{conjecture}}
\def\bcr{\begin{corollary}}
	\def\ecr{\end{corollary}}
\def\bd{\begin{definition}}
	\def\ed{\end{definition}}
\def\bea{\begin{eqnarray}}
\def\eea{\end{eqnarray}}
\def\bem{\begin{enumerate}}
	\def\eem{\end{enumerate}}
\def\bex{\begin{example}}
	\def\eex{\end{example}}
\def\bim{\begin{itemize}}
	\def\eim{\end{itemize}}
\def\bl{\begin{lemma}}
	\def\el{\end{lemma}}
\def\bma{\begin{bmatrix}}
	\def\ema{\end{bmatrix}}
\def\bpf{\begin{proof}}
	\def\epf{\end{proof}}
\def\bpp{\begin{proposition}}
	\def\epp{\end{proposition}}
\def\bqu{\begin{question}}
	\def\equ{\end{question}}
\def\br{\begin{remark}}
	\def\er{\end{remark}}
\def\bt{\begin{theorem}}
	\def\et{\end{theorem}}

\def\squareforqed{\hbox{\rlap{$\sqcap$}$\sqcup$}}
\def\qed{\ifmmode\squareforqed\else{\unskip\nobreak\hfil
		\penalty50\hskip1em\null\nobreak\hfil\squareforqed
		\parfillskip=0pt\finalhyphendemerits=0\endgraf}\fi}
\def\endenv{\ifmmode\;\else{\unskip\nobreak\hfil
		\penalty50\hskip1em\null\nobreak\hfil\;
		\parfillskip=0pt\finalhyphendemerits=0\endgraf}\fi}
\newenvironment{proof}{\noindent \textbf{{Proof.~} }}{\qed}
\def\Dbar{\leavevmode\lower.6ex\hbox to 0pt
	{\hskip-.23ex\accent"16\hss}D}
\makeatletter
\def\url@leostyle{%
	\@ifundefined{selectfont}{\def\UrlFont{\sf}}{\def\UrlFont{\small\ttfamily}}}
\makeatother
\def\bcj{\begin{conjecture}}
	\def\ecj{\end{conjecture}}
\def\bcr{\begin{corollary}}
	\def\ecr{\end{corollary}}
\def\bd{\begin{definition}}
	\def\ed{\end{definition}}
\def\bea{\begin{eqnarray}}
\def\eea{\end{eqnarray}}
\def\bem{\begin{enumerate}}
	\def\eem{\end{enumerate}}
\def\bex{\begin{example}}
	\def\eex{\end{example}}
\def\bim{\begin{itemize}}
	\def\eim{\end{itemize}}
\def\bl{\begin{lemma}}
	\def\el{\end{lemma}}
\def\bpf{\begin{proof}}
	\def\epf{\end{proof}}
\def\bpp{\begin{proposition}}
	\def\epp{\end{proposition}}
\def\bqu{\begin{question}}
	\def\equ{\end{question}}
\def\br{\begin{remark}}
	\def\er{\end{remark}}
\def\bt{\begin{theorem}}
	\def\et{\end{theorem}}

\def\btb{\begin{tabular}}
	\def\etb{\end{tabular}}

\newcommand{\nc}{\newcommand}

\nc{\bbA}{\mathbb{A}} \nc{\bbB}{\mathbb{B}} \nc{\bbC}{\mathbb{C}}
\nc{\bbD}{\mathbb{D}} \nc{\bbE}{\mathbb{E}} \nc{\bbF}{\mathbb{F}}
\nc{\bbG}{\mathbb{G}} \nc{\bbH}{\mathbb{H}} \nc{\bbI}{\mathbb{I}}
\nc{\bbJ}{\mathbb{J}} \nc{\bbK}{\mathbb{K}} \nc{\bbL}{\mathbb{L}}
\nc{\bbM}{\mathbb{M}} \nc{\bbN}{\mathbb{N}} \nc{\bbO}{\mathbb{O}}
\nc{\bbP}{\mathbb{P}} \nc{\bbQ}{\mathbb{Q}} \nc{\bbR}{\mathbb{R}}
\nc{\bbS}{\mathbb{S}} \nc{\bbT}{\mathbb{T}} \nc{\bbU}{\mathbb{U}}
\nc{\bbV}{\mathbb{V}} \nc{\bbW}{\mathbb{W}} \nc{\bbX}{\mathbb{X}}
\nc{\bbZ}{\mathbb{Z}}

\nc{\bA}{{\bf A}} \nc{\bB}{{\bf B}} \nc{\bC}{{\bf C}}
\nc{\bD}{{\bf D}} \nc{\bE}{{\bf E}} \nc{\bF}{{\bf F}}
\nc{\bG}{{\bf G}} \nc{\bH}{{\bf H}} \nc{\bI}{{\bf I}}
\nc{\bJ}{{\bf J}} \nc{\bK}{{\bf K}} \nc{\bL}{{\bf L}}
\nc{\bM}{{\bf M}} \nc{\bN}{{\bf N}} \nc{\bO}{{\bf O}}
\nc{\bP}{{\bf P}} \nc{\bQ}{{\bf Q}} \nc{\bR}{{\bf R}}
\nc{\bS}{{\bf S}} \nc{\bT}{{\bf T}} \nc{\bU}{{\bf U}}
\nc{\bV}{{\bf V}} \nc{\bW}{{\bf W}} \nc{\bX}{{\bf X}}
\nc{\bZ}{{\bf Z}} \nc{\bv}{{\bf v}}
\nc{\ba}{{\bf a}} \nc{\be}{{\bf e}} \nc{\bu}{{\bf u}}
\nc{\brr}{{\bf r}}

\nc{\cA}{{\cal A}} \nc{\cB}{{\cal B}} \nc{\cC}{{\cal C}}
\nc{\cD}{{\cal D}} \nc{\cE}{{\cal E}} \nc{\cF}{{\cal F}}
\nc{\cG}{{\cal G}} \nc{\cH}{{\cal H}} \nc{\cI}{{\cal I}}
\nc{\cJ}{{\cal J}} \nc{\cK}{{\cal K}} \nc{\cL}{{\cal L}}
\nc{\cM}{{\cal M}} \nc{\cN}{{\cal N}} \nc{\cO}{{\cal O}}
\nc{\cP}{{\cal P}} \nc{\cQ}{{\cal Q}} \nc{\cR}{{\cal R}}
\nc{\cS}{{\cal S}} \nc{\cT}{{\cal T}} \nc{\cU}{{\cal U}}
\nc{\cV}{{\cal V}} \nc{\cW}{{\cal W}} \nc{\cX}{{\cal X}}
\nc{\cZ}{{\cal Z}}

\nc{\hA}{{\hat{A}}} \nc{\hB}{{\hat{B}}} \nc{\hC}{{\hat{C}}}
\nc{\hD}{{\hat{D}}} \nc{\hE}{{\hat{E}}} \nc{\hF}{{\hat{F}}}
\nc{\hG}{{\hat{G}}} \nc{\hH}{{\hat{H}}} \nc{\hI}{{\hat{I}}}
\nc{\hJ}{{\hat{J}}} \nc{\hK}{{\hat{K}}} \nc{\hL}{{\hat{L}}}
\nc{\hM}{{\hat{M}}} \nc{\hN}{{\hat{N}}} \nc{\hO}{{\hat{O}}}
\nc{\hP}{{\hat{P}}} \nc{\hR}{{\hat{R}}} \nc{\hS}{{\hat{S}}}
\nc{\hT}{{\hat{T}}} \nc{\hU}{{\hat{U}}} \nc{\hV}{{\hat{V}}}
\nc{\hW}{{\hat{W}}} \nc{\hX}{{\hat{X}}} \nc{\hZ}{{\hat{Z}}}

\nc{\hn}{{\hat{n}}}

\def\ghz{\mathop{\rm GHZ}}

\def\rank{\mathop{\rm rank}}

\newcommand{\bra}[1]{\langle#1|}
\newcommand{\ket}[1]{|#1\rangle}

\newcommand{\ketbra}[2]{|#1\rangle\!\langle#2|}
\newcommand{\braket}[2]{\langle#1|#2\rangle}

\newcommand{\fl}[2]{\lfloor\frac{#1}{#2}\rfloor}

\def\Dbar{\leavevmode\lower.6ex\hbox to 0pt
	{\hskip-.23ex\accent"16\hss}D}

\begin{document}
\title{Strong quantum nonlocality in $N$-partite systems}		
	\author{Fei Shi}
\email[]{shifei@mail.ustc.edu.cn}
\affiliation{School of Cyber Security,
	University of Science and Technology of China, Hefei, 230026, People's Republic of China}

\author{Zuo Ye}
\email[]{zyprince@mail.ustc.edu.cn}
\affiliation{School of Mathematical Sciences,
 University of Science and Technology of China, Hefei, 230026, People's Republic of China}

\author{Lin Chen}
\email[]{linchen@buaa.edu.cn}
\affiliation{LMIB(Beihang University), Ministry of Education, and School of Mathematical Sciences, Beihang University, Beijing 100191, China}
\affiliation{International Research Institute for Multidisciplinary Science, Beihang University, Beijing 100191, China}

\author{Xiande Zhang}
\email[]{Corresponding authors: drzhangx@ustc.edu.cn}
\affiliation{School of Mathematical Sciences,
	University of Science and Technology of China, Hefei, 230026, People's Republic of China}

\begin{abstract}
A set of multipartite orthogonal quantum states is strongly nonlocal if it is locally irreducible for every bipartition of the subsystems [\href{https://journals.aps.org/prl/abstract/10.1103/PhysRevLett.122.040403}{Phys. Rev. Lett. \textbf{122}, 040403 (2019)}].
Although this property has been shown in three-, four- and five-partite systems, the existence of strongly nonlocal sets in  $N$-partite systems remains unknown when $N\geq 6$. In this paper, we successfully show that a strongly nonlocal set of orthogonal entangled states exists in $(\bbC^d)^{\otimes N}$ for all $N\geq 3$ and $d\geq 2$, which for the first time reveals the strong quantum nonlocality in general $N$-partite systems. For $N=3$ or $4$ and  $d\geq 3$, we present a strongly nonlocal set consisting of genuinely entangled states, which has a smaller size than any known strongly nonlocal orthogonal product set.
Finally, we connect strong quantum nonlocality with local hiding of information as an application.
\end{abstract}

\maketitle

\section{Introduction}\label{sec:int}
Quantum nonlocality is one of the most important properties in quantum mechanics. The entangled states show Bell nonlocality for violating Bell-type inequalities \cite{horodecki2009quantum,brunner2014bell}.
A set of multipartite orthogonal quantum states is locally indistinguishable if it is not possible to optimally distinguish the states
by any sequence of local operations and classical communications (LOCC). This set also shows quantum nonlocality, which is different from Bell nonlocality.   Bennett \emph{et al.} firstly constructed a locally indistinguishable  orthogonal product basis (OPB) in  $\bbC^3\otimes\bbC^3$ \cite{bennett1999quantum},  which shows the phenomenon of quantum nonlocality without entanglement.
Later, locally indistinguishable orthogonal product sets (OPSs) and  orthogonal entangled sets (OESs) have been widely studied \cite{1,2,3,4,5,6,7,8,9,10,11,12,13,14,15,16,17}. When information is encoded in a locally indistinguishable set of a composite quantum system,  it cannot be completely retrieved under LOCC among the spatially separated subsystems.  Consequently, local indistinguishability can be used for quantum data hiding \cite{terhal2001hiding,divincenzo2002quantum,eggeling2002hiding,Matthews2009Distinguishability} and quantum secret sharing \cite{Markham2008Graph,Hillery,Rahaman}.

Recently, Halder \emph{et al.} introduced the concepts of local irreducibility and strong quantum nonlocality \cite{Halder2019Strong}. A set of multipartite orthogonal states is locally irreducible if it is not possible to eliminate one or more states from the set by orthogonality-preserving local measurements. A locally irreducible set must be a locally indistinguishable set,  while the converse is not true in general.  Further, a set of multipartite orthogonal states is strongly nonlocal if it is locally irreducible for each  bipartition of the subsystems.
The authors of Ref.~\cite{Halder2019Strong} showed the phenomenon of strong quantum nonlocality without entanglement, by presenting two strongly nonlocal OPBs  in $3\otimes 3\otimes 3$ and $4\otimes 4\otimes 4$, respectively.  Later, there were several results for strongly nonlocal OPSs and OESs. For OPSs,  Yuan \emph{et al.} presented some strongly nonlocal OPSs in $\bbC^d\otimes \bbC^d\otimes \bbC^d$, $\bbC^d\otimes \bbC^d\otimes \bbC^{d+1}$, $\bbC^3\otimes \bbC^3\otimes\bbC^3\otimes \bbC^3$, and $\bbC^4\otimes \bbC^4\otimes\bbC^4\otimes \bbC^4$. \cite{yuan2020strong}.  Recently, strongly nonlocal OPSs in general three-, four-, and five-partite systems,  and  strongly nonlocal unextendible product bases (UPBs) in general three-, and four-partite systems   were constructed  \cite{shi2021hyper,shi2021,shi2021upb}. For OESs, based on the Rubik’s cube, Shi \emph{et al.} presented  some strongly nonlocal OESs in $\bbC^d\otimes \bbC^d\otimes \bbC^d$  \cite{2020Strong} , while these OESs are not orthogonal genuinely entangled sets (OGESs). By using graph connectivity, Wang \emph{et al.}    successfully constructed  strongly nonlocal OGESs in $\bbC^d\otimes \bbC^d\otimes \bbC^d$ \cite{li2}.  The concept of strong quantum nonlocality was also extended to more general settings \cite{zhangstrong2019,Sumit2019Genuinely,halder2020dis,li2020local,zhang2020strong}.

Whether we use OPS or OES, the existence of strongly nonlocal sets are limited to three-, four-, and five-partite systems up to now. It is still unknown whether strong quantum nonlocality can be shown in general $N$-partite systems. So it is natural to ask the following question.
\begin{itemize}
  \item[Q:] Can we show the strong quantum nonlocality  in $(\bbC^d)^{\otimes N}$ for $d\geq 2$ and for all $N\geq 3$?
\end{itemize}


The main method of showing the strong nonlocality of a set is to show a stronger property, that is, any orthogonality-preserving local measurement that may be performed across any
bipartition of the subsystems must be trivial.
Although it suffices to show that the condition holds for any part consisting of $N-1$ subsystems \cite{shi2021,li2}, it is not easy to prove that any joint orthogonality-preserving local measurement performed on any set of $N-1$ subsystems must be trivial
when $N$ is  large. To reduce the complexity,  we wish to construct a set of orthogonal states which has a similar structure when restricted on any $N-1$ subsystems. However, it is again not easy to give an explicit form of the states satisfying this condition, which is important in the verification of triviality of the measurement.

One of the main contribution of this paper is to answer the question Q in an affirmative way. In fact, by using cyclic permutation group action, we construct a strongly nonlocal OES of size $d^N-(d-1)^N+1$  in  $(\bbC^d)^{\otimes N}$ for all $d\geq 2$ and $N\geq 3$. Based on this construction, we further show that when $N=3$ and $4$, a strongly nonlocal OGES of size $d^N-(d-1)^N+1$  exists in $(\bbC^d)^{\otimes N}$ for all $d\geq 2$. In Ref.~\cite{li2}, the authors asked whether we can construct a strongly nonlocal set via OGES that has a smaller size  than that via OPS in the same system.
The question was raised due to the intuition that OGESs show more strong nonlocality than OPSs \cite{Halder2019Strong}, but known strongly nonlocal sets via OGESs  \cite{li2} have larger sizes than that via OPSs \cite{yuan2020strong}. When $N=3$, the size $3d^2-3d+2$ of OGES in our construction is about half of the size $6(d-1)^2$ of the OPS in Ref. \cite{yuan2020strong}.
Finally, for applications, we show that strong quantum nonlocality can be used for local hiding of information.

The rest of this paper is organized as follows. In Sec.~\ref{sec:pre}, we introduce the concept of strong  nonlocality. In Sec.~\ref{sec:OESs}, we construct OESs in $(\bbC^d)^{\otimes N}$ for $d\geq 2$ and $N\geq 3$ from cyclic permutation group actions, and show the strong nonlocality  in Sec.~\ref{sec:strong_OES}. Next, in Sec.~\ref{sec:strong_OGES}, we present strongly nonlocal OGESs when $N=3$ and $4$. In Sec.~\ref{sec:hiding}, we connect strong quantum nonlocality with local hiding of information.  Finally, we conclude in Sec.~\ref{sec:con}.

\section{Preliminaries}\label{sec:pre}
	Throughout this paper, we only consider pure states, and we do not normalize states for simplicity. A positive operator-valued measure (POVM) on Hilbert space $\cH$ is a set of semidefinite operators $\{E_m=M_m^{\dagger}M_m\}$ such that  $\sum_m E_m=\bbI_{\cH}$, where each $E_m$ is called a POVM element, and  $\bbI_{\cH}$ is the identity operator on $\cH$. We only consider POVM measurements. A measurement is trivial if all its POVM elements are
	proportional to the identity operator. Otherwise, the measurement is called nontrivial.

For an integer $d\geq 2$,  we denote $\bbZ_d:=\{0,1,\ldots,d-1\}$,  $\bbZ_d^N:=\bbZ_d\times \bbZ_d\times \cdots\times\bbZ_d$ and $(\bbC^d)^{\otimes N}:=\bbC^d\otimes\bbC^d\otimes\cdots\otimes\bbC^d$, where $\bbZ_d$ and $\bbC^d$ both repeat $N$ times. We assume that $\{\ket{i}\}_{i\in\bbZ_d}$ is a computational basis of $\bbC^d$. For a bipartite state $\ket{\psi}_{AB}\in\bbC^m\otimes\bbC^n$, it can be expressed by
	\begin{equation}
	\ket{\psi}_{AB}=\sum_{i\in\bbZ_m,j\in\bbZ_n}a_{i,j}\ket{i}_A\ket{j}_B.
	\end{equation}
	Then $\ket{\psi}_{AB}$ corresponds to an $m\times n$ matrix $M=(a_{i,j})_{i\in\bbZ_m,j\in\bbZ_n}$. We denote
	\begin{equation}
	\rank(\ket{\psi}_{AB})=\rank(M),
	\end{equation}
	which is also called the \emph{Schmidt rank} of $\ket{\psi}_{AB}$ \cite{computation2010}. Then $\ket{\psi}_{AB}$ is an entangled state if and only if  $\rank(\ket{\psi}_{AB})\geq 2$. An $N$-partite state $\ket{\psi}_{A_1A_2\cdots A_N}$ is called an \emph{entangled state}, if it is entangled for at least one  bipartition of the subsystems $\{A_1,A_2,\ldots,A_N\}$. Moreover, an $N$-partite state $\ket{\psi}_{A_1A_2\cdots A_N}$ is called a \emph{genuinely entangled state}, if it is entangled for each bipartition of the subsystems $\{A_1,A_2,\ldots,A_N\}$.

	The most well known genuinely entangled states are $\ghz$ states and W states.
	An $N$-qudit $\ghz$ state  can be expressed by 		
	\begin{equation}
	\ket{\ghz}_d^N=\sum_{i\in\bbZ_d}\ket{i}_{A_1}\ket{i}_{A_2}\cdots\ket{i}_{A_N}.
	\end{equation}
	An $N$-qubit W state can be expressed by
	\begin{equation}
	\begin{aligned}
	\ket{W}_2^N=&\ket{1}_{A_1}\ket{0}_{A_2}\cdots \ket{0}_{A_N}+\ket{0}_{A_1}\ket{1}_{A_2}\cdots \ket{0}_{A_N}\\
	&+\cdots+\ket{0}_{A_1}\ket{0}_{A_2}\cdots \ket{1}_{A_N}.
	\end{aligned}
	\end{equation}

In local state discrimination, we usually perform orthogonality-preserving local measurements (OPLMs) for a set of orthogonal states, where a measurement is orthogonality preserving if the postmeasurement states keep being mutually orthogonal. Recently, Halder \emph{et al.} proposed the concepts of locally irreducible sets and strong quantum nonlocality \cite{Halder2019Strong}. A set of orthogonal states in $(\bbC^d)^{\otimes N}$ is \emph{locally irreducible} if it is not possible to eliminate one or more states from the set by OPLMs.  Moreover,  A set of orthogonal states in $(\bbC^d)^{\otimes N}$ is \emph{strongly nonlocal}, if it is locally irreducible for each bipartition of the subsystems. In fact, a stronger property is often applied to show  the strong nonlocality \cite{shi2021,li2}.
A set of orthogonal states  in $(\bbC^d)^{\otimes N}$  is said to have  the property of the strongest nonlocality when the following condition holds, any OPLM that may be performed across any
bipartition must be trivial. The following lemma shows that we only need to show that any $N-1$ parties can perform only a trivial OPLM \cite{shi2021}.

	\begin{lemma}\cite{shi2021}\label{lem:cyc}
	Let $\cS:=\{\ket{\psi_j}\}$ be a set of orthogonal states in a multipartite system $\otimes_{i=1}^{N}\cH_{A_i}$. For each $i=1,2,\ldots,N$, define $B_i=\{A_1A_2\ldots A_N\}\setminus \{A_i\}$ be the joint party of all but the $i$th party. Then the set  $\cS$  has  the property of the strongest nonlocality if the following condition holds for any $1\leq i\leq N$: if party $B_i$  performs any OPLM, then the OPLM is trivial.
   \end{lemma}

	\section{OESs from cyclic permutation group action}\label{sec:OESs}
	In this section, we construct OESs from cyclic permutation group action. We first briefly recall the concept and properties of group action \cite{rotman2010advanced}. If $X$ is a set and $G$ is a group, then $G$ $\emph{acts}$ on $X$ if there is a function $G\times X\rightarrow X$, denoted by $(g,x)\rightarrowtail gx$, such that
	\begin{enumerate}[(i)]
		\item $(gh)x=g(hx)$ for all $g,h\in G$ and $x\in X$;
		\item  $1x=x$ for all $x\in X$, where $1$ is the identity in $G$.
	\end{enumerate}
	For $x\in X$, the \emph{orbit} of $x$, denoted by $\cO_x$, is the subset
	\begin{equation}
	\cO_x=\{gx\mid g\in G\}\subseteq X,
	\end{equation} and $x$ is called a \emph{representative} of the orbit $\cO_x$.
	If $y\in \cO_x$, then $\cO_x=\cO_y$;  if $y\notin \cO_x$, then $\cO_x\cap\cO_y=\emptyset$. Thus, $X$ is the union of the mutually disjoint orbits,
	\begin{equation}\label{eq:orbit}
	X=\mathop{\bigcup}\limits_x\cO_x,
	\end{equation}
	where  $x$ runs over the set of representatives of all orbits. Moreover, $|\cO_x|$ divides $|G|$.

In our construction, the set we consider is the subset $\bbX_d^N$ of $\bbZ_d^N$,  which is defined by
	\begin{equation}
	\bbX_d^N:=\{(i_1,i_2,\ldots,i_N)\mid \mathop{\prod}\limits_{1\leq k\leq N} i_k=0\}.
	\end{equation}
	That is, for any $(i_1,i_2,\ldots,i_N)\in \bbX_d^N$, there exists at least one $i_k=0$ for $1\leq k\leq N$. Then $|\bbX_d^N|=d^N-(d-1)^N$.

 Assume that
 \begin{equation}
     G_N=\{\sigma^k\mid k\in\bbZ_N\}
 \end{equation} is a cyclic permutation group of order $N$, where
	\begin{equation}
	\sigma(i_1,i_2,\ldots,i_N)=(i_2,\ldots,i_N,i_1)
	\end{equation}
	for an $N$-tuple $(i_1,i_2,\ldots,i_N)\in \bbZ_d^N$.  Then $G_N$ acts on $\bbX_d^N$ by definition, and yields a partition of $\bbX_d^N$ into disjoint orbits.

	For example, since $G_3$ acts on  $\bbX_2^3=\bbZ_2^3\setminus\{(1,1,1)\}$, then
	\begin{equation}\label{eq:orbit_222}
	\bbX_2^3=\cO_{(0,0,0)}\cup\cO_{(0,0,1)}\cup\cO_{(0,1,1)},
	\end{equation} where
	\begin{equation}
	\begin{aligned}
	\cO_{(0,0,0)}&=\{(0,0,0)\},\\
	\cO_{(0,0,1)}&=\{(0,0,1),(0,1,0),(1,0,0)\},\\
	\cO_{(0,1,1)}&=\{(0,1,1),(1,1,0),(1,0,1)\}.\\
	\end{aligned}
	\end{equation}

\begin{figure}[t]
	\centering
	\includegraphics[scale=1]{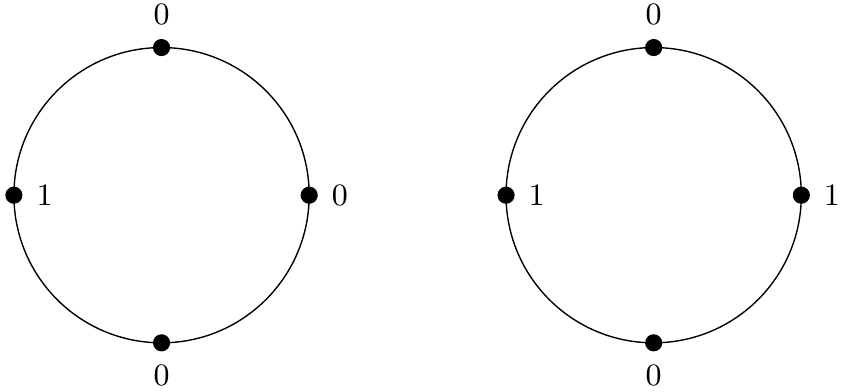}
	\caption{We represent two orbits by two cycles. For the left cycle, the orbit is $\cO_{(0,0,0,1)}=\{(0,0,0,1),$ $(0,0,1,0),(0,1,0,0),(1,0,0,0)\}$, which yields the set of states $\cS_{(0,0,0,1)}=\{\ket{0}_{A_1}\ket{0}_{A_2}\ket{0}_{A_3}\ket{1}_{A_4}+
	w_4^{s}\ket{0}_{A_1}\ket{0}_{A_2}\ket{1}_{A_3}\ket{0}_{A_4}
	+w_4^{2s}\ket{0}_{A_1}\ket{1}_{A_2}\ket{0}_{A_3}\ket{0}_{A_4}+
	w_4^{3s}\ket{1}_{A_1}\ket{0}_{A_2}\ket{0}_{A_3}\ket{0}_{A_4}
	: s\in\bbZ_4\}$.  For the right cycle, the orbit is $\cO_{(0,1,0,1)}=\{(0,1,0,1),(1,0,1,0)\}$,  which yields the set of states  $\cS_{(0,1,0,1)}=\{\ket{0}_{A_1}\ket{1}_{A_2}\ket{0}_{A_3}\ket{1}_{A_4}\pm\ket{1}_{A_1}\ket{0}_{A_2}\ket{1}_{A_3}\ket{0}_{A_4}\}.$} \label{fig:cycle}
\end{figure}

	In general if $x\neq (0,0\ldots,0)$, then we can write
	\begin{equation}\label{eq:O_x}
	\cO_{x}=\{(i_1^{(j)},i_2^{(j)},\ldots,i_N^{(j)})\mid j\in\bbZ_k\},
	\end{equation}
where $k\geq 2$ is the size of  $\cO_x$. Note that the size  may be different for different orbits, see  Fig.~\ref{fig:cycle} as an example.

For each orbit $\cO_x$ of $\bbX_d^N$,
  we define a set $\cS_x$ of states in $(\bbC^d)^{\otimes N}$ as follows. If $x\neq (0,0\ldots,0)$, let
	\begin{equation}\label{eq:S_x}
	\cS_x:=\{\sum_{j\in\bbZ_k}w_k^{sj}\ket{i_1^{(j)}}_{A_1}\ket{i_2^{(j)}}_{A_2}\cdots\ket{i_N^{(j)}}_{A_N}: s\in\bbZ_k \},
	\end{equation}
	where $w_k=e^{\frac{2\pi\sqrt{-1}}{k}}$ is a primitive $k$th root of unity. In fact, the coefficient matrix
	\begin{equation}\label{eq:B}
B=(w_k^{sj})_{s,j\in\bbZ_k}
\end{equation}
 forms a complex Hadamard matrix of order $k$. By definition, the number of states in $\cS_x$ equals the size of $\cO_x$ in this case. See  Fig.~\ref{fig:cycle} for $\cS_{(0,0,0,1)}$ and $\cS_{(0,1,0,1)}$.
If $\cO_x=\{(0,0,\ldots,0)\}$, we define
	\begin{equation}\label{eq:S_0}
	\cS_{(0,0,\ldots,0)}:=\{\ket{0}_{A_1}\ket{0}_{A_1}\cdots\ket{0}_{A_N}\pm\ket{1}_{A_1}\ket{1}_{A_1}\cdots\ket{1}_{A_N}\},
	\end{equation}which has one more element than the orbit.  Since $\bbX_d^N=\mathop{\bigcup}\limits_x\cO_x$ is a disjoint union, thus
	\begin{equation}\label{eq:dN}
	\cB_d^N:=\mathop{\bigcup}\limits_x \cS_x
	\end{equation} is also a disjoint union, which has size $d^N-(d-1)^N+1$.

	For example, in $(\bbC^2)^{\otimes 3}$, we know that
	\begin{equation}\label{eq:222}
	\cB_2^3= \cS_{(0,0,0)}\cup\cS_{(0,0,1)}\cup \cS_{(0,1,1)},
	\end{equation} by Eq.~\eqref{eq:orbit_222}, where
	\begin{equation}
	\begin{aligned}
	S_{(0,0,0)}=\{&\ket{0}_{A_1}\ket{0}_{A_2}\ket{0}_{A_3}\pm\ket{1}_{A_1}\ket{1}_{A_2}\ket{1}_{A_3}\},\\
	S_{(0,0,1)}=\{&\ket{0}_{A_1}\ket{0}_{A_2}\ket{1}_{A_3}+w_3^{s}\ket{0}_{A_1}\ket{1}_{A_2}\ket{0}_{A_3}\\&+w_3^{2s}\ket{1}_{A_1}\ket{0}_{A_2}\ket{0}_{A_3}: s\in\bbZ_3\},\\
	S_{(0,1,1)}=\{&\ket{0}_{A_1}\ket{1}_{A_2}\ket{1}_{A_3}+w_3^{s}\ket{1}_{A_1}\ket{1}_{A_2}\ket{0}_{A_3}\\&+w_3^{2s}\ket{1}_{A_1}\ket{0}_{A_2}\ket{1}_{A_3}: s\in\bbZ_3\},\\
	\end{aligned}
	\end{equation}
	and  $|\cB_2^3|=8$. One can easily check that  $\cB_2^3$ is an OES in $\bbC^2\otimes\bbC^2\otimes\bbC^2$.
	
   Similarly, we can construct $\cB_2^4$ in $\bbC^2\otimes\bbC^2\otimes\bbC^2\otimes\bbC^2$.
   Since
   \begin{equation}
   \begin{aligned}
   \bbX_2^4=&\cO_{(0,0,0,0)}\cup\cO_{(0,0,0,1)}\cup\cO_{(0,0,1,1)}\cup\cO_{(0,1,0,1)}\\
   &\cup\cO_{(0,1,1,1)},
         \end{aligned}
   \end{equation}
   we obtain that
      \begin{equation}\label{eq:B_24}
      \begin{aligned}
       \cB_2^4=&\cS_{(0,0,0,0)}\cup\cS_{(0,0,0,1)}\cup\cS_{(0,0,1,1)}\cup\cS_{(0,1,0,1)}\\
       &\cup\cS_{(0,1,1,1)},
    \end{aligned}
   \end{equation}
   where
   \begin{equation}
	\begin{aligned}
	S_{(0,0,0,0)}=\{&\ket{0}_{A_1}\ket{0}_{A_2}\ket{0}_{A_3}\ket{0}_{A_4}\\&
	\pm\ket{1}_{A_1}\ket{1}_{A_2}\ket{1}_{A_3}\ket{1}_{A_4}\},\\
	S_{(0,0,0,1)}=\{&\ket{0}_{A_1}\ket{0}_{A_2}\ket{0}_{A_3}\ket{1}_{A_4}\\&+w_4^{s}\ket{0}_{A_1}\ket{0}_{A_2}\ket{1}_{A_3}\ket{0}_{A_4}\\&+w_4^{2s}\ket{0}_{A_1}\ket{1}_{A_2}\ket{0}_{A_3}\ket{0}_{A_4}\\&+w_4^{3s}\ket{1}_{A_1}\ket{0}_{A_2}\ket{0}_{A_3}\ket{0}_{A_4}: s\in\bbZ_4\},\\
	S_{(0,0,1,1)}=\{&\ket{0}_{A_1}\ket{0}_{A_2}\ket{1}_{A_3}\ket{1}_{A_4}\\&+w_4^{s}\ket{0}_{A_1}\ket{1}_{A_2}\ket{1}_{A_3}\ket{0}_{A_4}\\&+w_4^{2s}\ket{1}_{A_1}\ket{1}_{A_2}\ket{0}_{A_3}\ket{0}_{A_4}\\&+w_4^{3s}\ket{1}_{A_1}\ket{0}_{A_2}\ket{0}_{A_3}\ket{1}_{A_4}: s\in\bbZ_4\},\\
	S_{(0,1,0,1)}=\{&\ket{0}_{A_1}\ket{1}_{A_2}\ket{0}_{A_3}\ket{1}_{A_4}\\&
	\pm\ket{1}_{A_1}\ket{0}_{A_2}\ket{1}_{A_3}\ket{0}_{A_4}\},\\
	S_{(0,1,1,1)}=\{&\ket{0}_{A_1}\ket{1}_{A_2}\ket{1}_{A_3}\ket{1}_{A_4}\\&+w_4^{s}\ket{1}_{A_1}\ket{1}_{A_2}\ket{1}_{A_3}\ket{0}_{A_4}\\&+w_4^{2s}\ket{1}_{A_1}\ket{1}_{A_2}\ket{0}_{A_3}\ket{1}_{A_4}\\&+w_4^{3s}\ket{1}_{A_1}\ket{0}_{A_2}\ket{1}_{A_3}\ket{1}_{A_4}: s\in\bbZ_4\},\\
	\end{aligned}
	\end{equation}	
	and $|\cB_2^4|=16$. We can also check that $\cB_2^4$ is an OES in $\bbC^2\otimes\bbC^2\otimes\bbC^2\otimes\bbC^2$.
	In general, we have the following result.
	
	\begin{lemma}\label{lem:B_d} The set $\cB_d^N$ is an OES of size $d^N-(d-1)^N+1$
		in $(\bbC^d)^{\otimes N}$.
	\end{lemma}

The proof of Lemma~\ref{lem:B_d} is given in Appendix~\ref{appendix:lem:B_d}.	
We will further show that $\cB_d^3$ is an OGES, while $\cB_d^4$ is not  in Section~\ref{sec:strong_OGES}. However, we first show that $\cB_d^N$ is strongly nonlocal  in the next section.

\section{The strong nonlocality for OESs in $N$-partite systems}\label{sec:strong_OES}
In this section, we show that the set $\cB_d^N$ of states in Eq.~\eqref{eq:dN} is strongly nonlocal  in $(\bbC^{d})^{\otimes N}$ for all $d\geq 2$ and $N\geq 3$. Since $\cB_d^N$ has a similar structure under the cyclic permutation of the subsystems $\{A_1,A_2,\ldots,A_N\}$, we only need to show that the party $A_2A_3\ldots A_N$ can only perform a trivial OPLM for $\cB_d^N$ by Lemma~\ref{lem:cyc}. We give Lemma~\ref{lemma:zero} and Lemma~\ref{lemma:trivial} in Appendix~\ref{appendix:two_lemmas}, which are useful for showing the strong nonlocality for OESs. 	 These  two lemmas can greatly reduce the complexity of proof when they are applied to $\{\cS_x\}$  of $\cB_d^N$.	First, we give two  examples to illustrate the idea of proof.	
	\begin{example}
		In $\bbC^2\otimes \bbC^2\otimes \bbC^2$, the set $\cB_2^3$ given by Eq.~\eqref{eq:222} is a strongly nonlocal OES of size $8$.
	\end{example}
	\begin{proof}
		Let $A_2$ and $A_3$ come together to  perform a joint  OPLM $\{E=M^{\dagger}M\}$, where each POVM
		element can be written as a $4\times 4$ matrix in the basis $\{\ket{i}_{A_2}\ket{j}_{A_3}\}_{i,j\in \bbZ_2}$,
		\begin{equation}
		E=\begin{pmatrix}
		a_{00,00} &a_{00,01} &a_{00,10}  &a_{00,11} \\
		a_{01,00} &a_{01,01} &a_{01,10}  &a_{01,11} \\
		a_{10,00} &a_{10,01} &a_{10,10}  &a_{10,11} \\
		a_{11,00} &a_{11,01} &a_{11,10}  &a_{11,11} \\
		\end{pmatrix}.
		\end{equation}
		Then the postmeasurement states $\{\bbI_{A_1}\otimes M\ket{\psi}:\ket{\psi}\in \cB_2^3\}$ should be mutually orthogonal. That is
		\begin{equation}
		\bra{\psi}\bbI_{A_1}\otimes E\ket{\phi}=0
		\end{equation}for $\ket{\psi}\neq \ket{\phi}\in\cB_2^3$. By using these orthogonality relations, we need to show that $E\propto \bbI$.
		
		Since $E=E^{\dagger}$, $a_{ij,k\ell}=0$ can imply that $a_{k\ell,ij}=0$. 	Applying Lemma~\ref{lemma:zero} to $\cS_{(0,0,0)}$ and $\cS_{(0,0,1)}$, we must have $a_{00,01}=a_{00,10}=a_{00,11}=0$. Next, applying Lemma~\ref{lemma:zero} to $\cS_{(0,0,0)}$ and  $\cS_{(0,1,1)}$, we obtain that $a_{10,11}=a_{01,11}=0$. Since $a_{00,01}=a_{00,10}=0$, we can apply Lemma~\ref{lemma:trivial} to $\cS_{(0,0,1)}$. Then we obtain that $a_{01,10}=0$ and $a_{01,01}=a_{10,10}=a_{00,00}$. Finally, applying Lemma~\ref{lemma:trivial} to $\cS_{(0,0,0)}$, we have $a_{00,00}=a_{11,11}$. Thus $E\propto \bbI$, and this completes the proof.
	\end{proof}	
	\vspace{0.4cm}

	\begin{example}
		In $\bbC^2\otimes \bbC^2\otimes \bbC^2\otimes\bbC^2$, the set $\cB_2^4$ given by Eq.~\eqref{eq:B_24} is a strongly nonlocal OES of size $16$.
	\end{example}
	\begin{proof}
		Let $A_2$, $A_3$ and $A_4$ come together to  perform a joint  OPLM $\{E=M^{\dagger}M\}$, where each POVM
		element can be written as a $8\times 8$ matrix in the basis $\{\ket{i}_{A_2}\ket{j}_{A_3}\ket{k}_{A_4}\}_{i,j,k\in \bbZ_2}$, $E=(a_{ijk,\ell mn})_{i,j,k,\ell,m,n\in\bbZ_2}$.
		Then the postmeasurement states $\{\bbI_{A_1}\otimes M\ket{\psi}:\ket{\psi}\in \cB_2^4\}$ should be mutually orthogonal. We can show that $E\propto \bbI$ from Table~\ref{Table:2222}. This completes this proof.
\end{proof}
\vspace{0.4cm}
		\begin{table*}[t]
	\renewcommand\arraystretch{1.7}	
	\caption{Off-diagonal elements and Diagonal elements of $E=(a_{ijk,\ell mn})_{i,j,k,\ell,m,n\in\bbZ_2}$. We apply Lemma~\ref{lemma:zero} to two sets among $\cS_{(0,0,0,0)},\cS_{(0,0,0,1)},\cS_{(0,0,1,1)},\cS_{(0,1,0,1)},\cS_{(0,1,1,1)}$,  and apply Lemma~\ref{lemma:trivial} to one set  among  $\cS_{(0,0,0,0)},\cS_{(0,0,0,1)},\cS_{(0,0,1,1)},\cS_{(0,1,0,1)},\cS_{(0,1,1,1)}$.    }\label{Table:2222}
	\centering
	\renewcommand\tabcolsep{5pt}
	\begin{tabular}{ll|ll}
		\midrule[1.1pt]
		Sets &Elements 	&Sets &Elements \\
		\midrule[1.1pt]
		$\cS_{(0,0,0,0)}$, $\cS_{(0,0,0,1)}$	  &$a_{000,001}=a_{000,010}=a_{000,100}=a_{000,111}=0$ 		&$\cS_{(0,0,0,0)}$, $\cS_{(0,1,1,1)}$   &$a_{110,111}=a_{101,111}=a_{011,111}=0$	\\
		
		$\cS_{(0,0,0,0)}$, $\cS_{(0,0,1,1)}$   &$a_{000,011}=a_{000,110}=a_{100,111}=a_{001,111}=0$   	
		&	$\cS_{(0,0,0,0)}$, $\cS_{(0,1,0,1)}$   &$a_{000,101}=a_{010,111}=0$    \\
		
	$\cS_{(0,0,1,1)}$, $\cS_{(0,1,0,1)}$   &$a_{011,101}=a_{101,110}=a_{010,100}=a_{001,010}=0$&$\cS_{(0,0,0,1)}$, $\cS_{(0,1,0,1)}$   &$a_{001,101}=a_{010,101}=a_{100,101}=0$   \\
		
		$\cS_{(0,0,0,1)}$, $\cS_{(0,0,1,1)}$   &$a_{001,011}=a_{001,110}=a_{010,011}$ 	&$	\cS_{(0,0,1,1)}$   &$a_{001,100}=0$\\

 	 &$=a_{010,110}=a_{011,100}=a_{100,110}=0$ &$\cS_{(0,0,1,1)}$   &$a_{011,110}=0$\\
		\midrule[1.1pt]
		$\cS_{(0,0,0,0)}$   &$a_{000,000}=a_{111,111}$  &	$\cS_{(0,0,1,1)}$   &$a_{100,100}=a_{011,011}=a_{110,110}$ \\
		
		$\cS_{(0,0,0,1)}$   &$a_{000,000}=a_{001,001}=a_{010,010}=a_{100,100}$		&$\cS_{(0,1,0,1)}$ &$a_{010,010}=a_{101,101}$\\
		\midrule[1.1pt]
	\end{tabular}
\end{table*}

	In fact, we can show a general result.

	\begin{theorem}\label{thm:OES}The set $\cB_d^N$ defined in Eq.~\eqref{eq:dN}
		is a strongly nonlocal OES of size $d^N-(d-1)^N+1$
		in $(\bbC^d)^{\otimes N}$ for all $d\geq 2$ and $N\geq 3$.
	\end{theorem}

The proof of Theorem~\ref{thm:OES} is given in Appendix~\ref{appendix:thm}. Note that Theorem~\ref{thm:OES} presents the first class of strongly nonlocal OESs in $(\bbC^d)^{\otimes N}$ for any $d\geq 2$ and $N\geq 3$.

\section{The strong nonlocality for small OGESs in $3,4$-partite systems}\label{sec:strong_OGES}
	
	In this section, we consider strongly nonlocal OGESs in $3,4$-partite systems which have smaller sizes than those of previously known.  Two states $\ket{\psi}$ and $\ket{\phi}$ are called LU-equivalent if there exists a product $U_1\otimes U_2\otimes \cdots\otimes U_N$ of unitary operators, such that
	\begin{equation}
	\ket{\psi}=U_1\otimes U_2\otimes \cdots\otimes U_N\ket{\phi}.
	\end{equation}
	It is known that LU-equivalence does not change the entanglement of a state.
	
		First, we consider $\cB_d^3$ in $\bbC^d\otimes \bbC^d\otimes \bbC^d$ for $d\geq 2$. Denote $[d-1]:=\{1,2,\ldots,d-1\}$. Then we can explicitly write down all states in $\cB_d^3$ as follows.
	Since
	\begin{equation}
	\begin{aligned}
	\bbX_d^3=&\cO_{(0,0,0)}\bigcup\left(\mathop{\cup}\limits_{i\in [d-1]}(\cO_{(0,0,i)}\cup\cO_{(0,i,i)})\right)\\&\bigcup \left(\mathop{\cup}\limits_{ i\neq  j\in [d-1]}\cO_{(0,i,j)}\right),
		\end{aligned}
	\end{equation}
	we obtain that
	\begin{equation}\label{eq:ddd}
	\begin{aligned}
	\cB_d^3=&\cS_{(0,0,0)}\bigcup\left(\mathop{\cup}\limits_{i\in [d-1]}(S_{(0,0,i)}\cup\cS_{(0,i,i)})\right)\\&\bigcup \left(\mathop{\cup}\limits_{ i\neq  j\in [d-1]}S_{(0,i,j)}\right),
		\end{aligned}
	\end{equation}
	where
	\begin{equation}
	\begin{aligned}
	\cS_{(0,0,0)}=\{&\ket{0}_{A_1}\ket{0}_{A_2}\ket{0}_{A_3}\pm\ket{1}_{A_1}\ket{1}_{A_2}\ket{1}_{A_3}\},\\
	\cS_{(0,0,i)}=\{&\ket{0}_{A_1}\ket{0}_{A_2}\ket{i}_{A_3}+w_3^s \ket{0}_{A_1}\ket{i}_{A_2}\ket{0}_{A_3}\\&+w_3^{2s} \ket{i}_{A_1}\ket{0}_{A_2}\ket{0}_{A_3}: s\in\bbZ_3\},\\
	\cS_{(0,i,i)}=\{&\ket{0}_{A_1}\ket{i}_{A_2}\ket{i}_{A_3}+w_3^s \ket{i}_{A_1}\ket{i}_{A_2}\ket{0}_{A_3}\\&+w_3^{2s} \ket{i}_{A_1}\ket{0}_{A_2}\ket{i}_{A_3}: s\in\bbZ_3\},\\
	\cS_{(0,i,j)}=\{&\ket{0}_{A_1}\ket{i}_{A_2}\ket{j}_{A_3}+w_3^s  \ket{i}_{A_1}\ket{j}_{A_2}\ket{0}_{A_3}\\&+w_3^{2s} \ket{j}_{A_1}\ket{0}_{A_2}\ket{i}_{A_3}: s\in\bbZ_3\},\\
	\end{aligned}
	\end{equation}
	for $ i\neq j\in[d-1]$.

	For $i\in[d-1]$ and $s\in\bbZ_3$, let
		\begin{equation*}
		\begin{aligned}
		\ket{\psi_s}=&\ket{0}_{A_1}\ket{0}_{A_2}\ket{i}_{A_3}+w_3^s\ket{0}_{A_1}\ket{i}_{A_2}\ket{0}_{A_3}\\&+     w_3^{2s}\ket{i}_{A_1}\ket{0}_{A_2}\ket{0}_{A_3}\in\cS(0,0,i).
		\end{aligned}
		\end{equation*}
		Then there must exist three unitary operators
		\begin{equation}
		  \begin{aligned}
		 P_1: \ket{0}_{A_1}\rightarrow \ket{0}_{A_1},  \ &\ket{i}_{A_1}\rightarrow w_3^{-2s}\ket{1}_{A_1};\\
		P_2: \ket{0}_{A_2}\rightarrow \ket{0}_{A_2}, \ &\ket{i}_{A_2}\rightarrow w_3^{-s}\ket{1}_{A_2};\\
		P_3: \ket{0}_{A_3}\rightarrow \ket{0}_{A_3}, \ &\ket{i}_{A_3}\rightarrow   \ket{1}_{A_3},
				  \end{aligned}
			\end{equation}
		such that
		\begin{equation*}
		\ket{W}_2^3=P_1\otimes P_2\otimes P_3	\ket{\psi_s}.
		\end{equation*}
		Thus $\ket{\psi_s}$ is LU-equivalent to a W-state. In the same way, we can show that any state in $\cS(0,0,0)$ and $\cS(0,i,j)$  is LU-equivalent to a $\ghz$-state, and any state in $\cS(0,i,i)$ is LU-equivalent to a W-state. Thus, we have the following result.	
	
  		\begin{table*}[t]
	\renewcommand\arraystretch{1.7}	
	\caption{Comparisons of the sizes between strongly nonlocal OPSs and strongly nonlocal OGESs.
 }\label{Table:compasions}
	\centering
	\renewcommand\tabcolsep{5pt}
	\begin{tabular}{l|ll|ll}
		\midrule[1.1pt]
		Systems &Sizes of the OPSs 	&References &Sizes of the OGESs 	&References \\
		\midrule[1.1pt]
	   $\bbC^3\otimes\bbC^3\otimes\bbC^3$ &$19$ &\cite{shi2021} &$18$ & Appendix~\ref{appendix:OGES_333}\\
	   	\hline
	   $\bbC^d\otimes\bbC^d\otimes\bbC^d$, $d\geq 3$ &$6(d-1)^2$ &\cite{yuan2020strong} &$d^3-(d-1)^3+1$ & Lemma~\ref{pro:ddd_OGES}\\
	   	\hline
	    $\bbC^d\otimes\bbC^d\otimes\bbC^d\otimes \bbC^d$, $d\geq 3$ &$d^{4}-(d-2)^4$ &\cite{shi2021hyper} &$d^4-(d-1)^4+1$ &Lemma~\ref{pro:OGES_dddd}\\
		\midrule[1.1pt]
	\end{tabular}
\end{table*}		
	
	\begin{lemma}\label{pro:ddd_OGES}
		For all $d\geq 2$, the set $\cB_d^3$ given by Eq.~\eqref{eq:ddd}
		is an OGES in $\bbC^d\otimes \bbC^d\otimes \bbC^d$.
	\end{lemma}

 Combining Theorem~\ref{thm:OES} and Lemma~\ref{pro:ddd_OGES}, we have presented a strongly nonlocal OGES of size $3d^2-3d+2$ in  $\bbC^d\otimes \bbC^d\otimes \bbC^d$ for all $d\geq 2$. When $d\geq 3$, the size of the strongly nonlocal OGES in our construction is about half of the size $6(d-1)^2$ of the strongly nonlocal OPS in Ref. \cite{yuan2020strong}. However, when $d=3$, the authors in Ref.~\cite{shi2021} presented a strongly nonlocal UPB of size $19$, which is smaller than the size $20$ of OGES in Lemma~\ref{pro:ddd_OGES}. For completeness, we construct a strongly nonlocal OGES of size $18$ in $\bbC^3\otimes\bbC^3\otimes\bbC^3$, see Appendix~\ref{appendix:OGES_333}.
See also Table~\ref{Table:compasions} for the comparisons.

It is natural to ask whether Lemma~\ref{pro:ddd_OGES} can be extended to $\cB_d^N$ with $N\geq 4$.
Unfortunately, the answer is no for $\cB_d^N$ when $N=4$.  In fact, the following state
	\begin{equation}
	\begin{aligned} \ket{\psi_i}=&\ket{0}_{A_1}\ket{0}_{A_2}\ket{i}_{A_3}\ket{i}_{A_4}+\ket{0}_{A_1}\ket{i}_{A_2}\ket{i}_{A_3}\ket{0}_{A_4}\\&+\ket{i}_{A_1}\ket{i}_{A_2}\ket{0}_{A_3}\ket{0}_{A_4}\\&+\ket{i}_{A_1}\ket{0}_{A_2}\ket{0}_{A_3}\ket{i}_{A_4}\in \cS_{(0,0,i,i)}\subset\cB_d^4
	\end{aligned}
	\end{equation}
is  separable in $A_1A_3|A_2A_4$ bipartition for each $i\in[d-1]$, since it can be written as
	\begin{equation}
	\ket{\psi_i}=(\ket{0}\ket{i}+\ket{i}\ket{0})_{A_1A_3}(\ket{0}\ket{i}+\ket{i}\ket{0})_{A_2A_4}.
	\end{equation}
 Thus $\cB_d^4$ is not an OGES. However, we can modify the coefficient matrix $B$ in Eq.~\eqref{eq:B} to get genuinely entangled states.

 Let 	\begin{equation}
	   \widetilde{B}=(b_{i,j})_{i,j\in\bbZ_4}:=\begin{pmatrix}
	   1 &1 &1 &2\\
	   1 &-1 &2 &-1\\
	   5 &5 &-2 &-4\\
	   5 &-5 &-4 &2\\
	   \end{pmatrix},
	\end{equation} which is a row orthogonal matrix.	For each $i\in[d-1]$, define
	\begin{equation}\label{eq:s00ii}
	\begin{aligned} \overline{\cS_{(0,0,i,i)}}:=&\{b_{s,0}\ket{0}_{A_1}\ket{0}_{A_2}\ket{i}_{A_3}\ket{i}_{A_4}\\&+b_{s,1}\ket{0}_{A_1}\ket{i}_{A_2}\ket{i}_{A_3}\ket{0}_{A_4}\\&+b_{s,2}\ket{i}_{A_1}\ket{i}_{A_2}\ket{0}_{A_3}\ket{0}_{A_4}\\
	&+b_{s,3}\ket{i}_{A_1}\ket{0}_{A_2}\ket{0}_{A_3}\ket{i}_{A_4}: s\in\bbZ_4\}.
	\end{aligned}
	\end{equation}
Let $\overline{\cB_d^4}$ be the set obtained from $\cB_d^4$ by replacing $\cS(0,0,i,i)$ with $\overline{\cS(0,0,i,i)}$ for all $i\in[d-1]$. Similar to Lemma~\ref{pro:ddd_OGES}, we can show that any state in $\cS_{(0,0,0,0)}$ and $\cS_{(0,i,0,i)}$ is LU-equivalent to a $\ghz$ state for $i\in [d-1]$, and any state in $\cS_{(0,0,0,i)}$ is LU-equivalent to a W state for $i\in [d-1]$. It is not clear for other states in $\overline{\cB_d^4}$. However, we are still able to show the genuine entanglement for all other states in $\overline{\cB_d^4}$.

	\begin{lemma}\label{pro:OGES_dddd}
For all $d\geq 2$, the set $\overline{\cB_d^4}$ is a strongly nonlocal OGES of size $d^4-(d-1)^4+1$ in $\bbC^d\otimes \bbC^d\otimes \bbC^d\otimes \bbC^d$.
	\end{lemma}

The proof of Lemma~\ref{pro:OGES_dddd} is given in Appendix~\ref{appendix:OGES_dddd}.	Lemma~\ref{pro:OGES_dddd} provides a strongly nonlocal OGES in four-partite systems whose size is smaller than that of the strongly nonlocal OPS in Ref.~\cite{shi2021hyper}. See also Table~\ref{Table:compasions} for a comparison.

Before closing this section, we remark that it is not easy to extend Lemma~\ref{pro:ddd_OGES} or Lemma~\ref{pro:OGES_dddd} to $N$-partite systems with $N\geq 5$. From the proof of Lemma~\ref{lem:B_d}, we know that for any state $\ket{\psi}\in\cB_d^N$, $\ket{\psi}$ is entangled in the bipartition $A_1|A_2A_3\cdots A_N$.
Since $\cB_d^N$ has a similar structure under the cyclic permutation of the subsystems $\{A_1,A_2,\ldots,A_N\}$,  $\ket{\psi}$ must be entangled in the bipartitions $\{A_1A_2\ldots A_N\}\setminus \{A_i\}$ for all $i\in [d]$. However, this is far from enough since we need to the entangled property for all bipartitions.

\begin{figure}[b]
	\centering
	\includegraphics[scale=0.6]{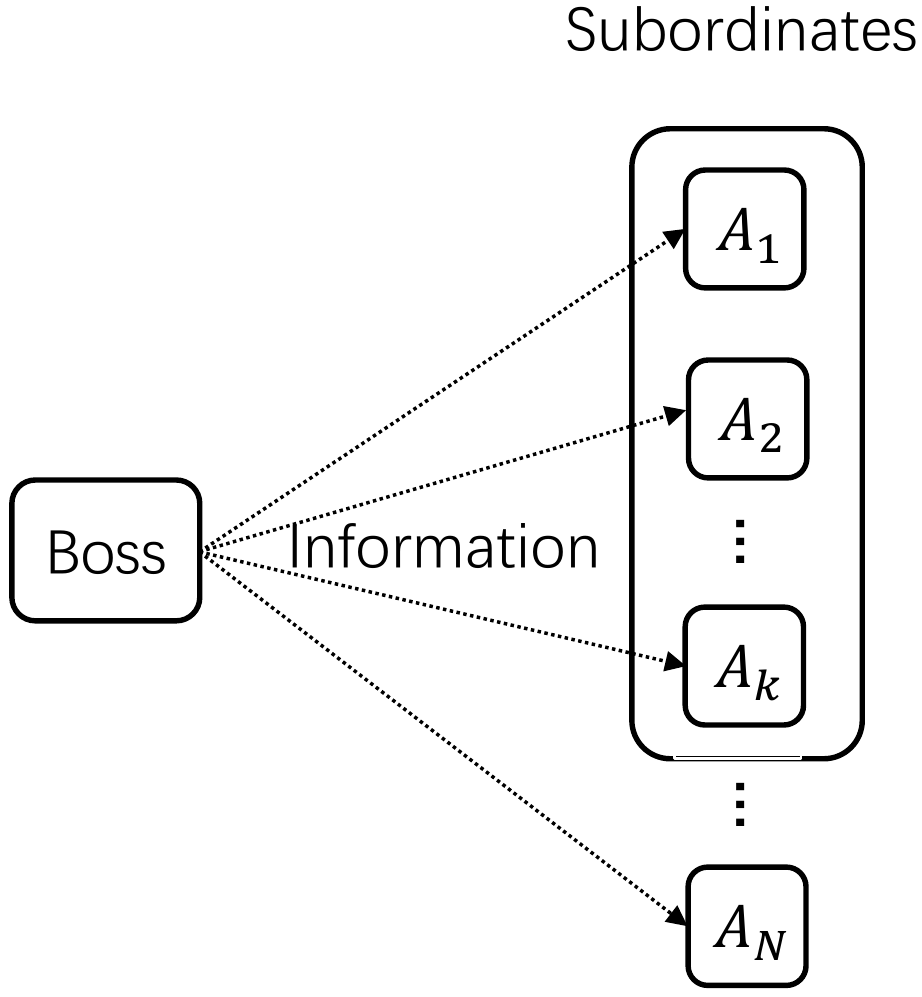}
	\caption{Local hiding of information. The information is encoded in a strongly nonlocal set of $N$-partite orthogonal states. The so called ``boss" sends the information to his  $N$ ``subordinates", $A_1, A_2, \ldots ,A_N$ through some quantum channels. If any $k$ ``subordinates" are collusive, $k<N$, and they can only perform OPLMs,   then the information will not be accessible at all.}\label{Fig:hiding}
\end{figure}

 \section{Local hiding of information} \label{sec:hiding}

  In this section, we indicate that strong quantum nonlocality can be used for local hiding of information \cite{bandyopadhyay2021activation,terhal2001hiding,divincenzo2002quantum,eggeling2002hiding,Matthews2009Distinguishability}. Assume that information is encoded in a locally indistinguishable set of $N$-partite orthogonal states,  and the so-called ``boss'' sends it to his $N$ ``subordinates''. These subordinates are from different labs, and they can only  communicate classically with each other and perform local measurements. Due to the local indistinguishability of the set, the information cannot be fully accessed by the subordinates, and part of it remains hidden. For example, since Bell basis is locally indistinguishable, one can encode 2 classical bits in four Bell states, but only 1 bit can be extracted locally \cite{terhal2001hiding,bandyopadhyay2021activation}. Note that if some subordinates are collusive, that is some subordinates are from the same lab and they can perform joint measurements,  then the information could be fully accessed by the $N$ subordinates. For example, when the information is encoded in a three-qubit UPB, then the information could be fully accessed by the three subordinates if any two of them are collusive \cite{bennett1999unextendible,1,cohen2008understanding}.  In order to avoid this collusion,  the boss can encode the information in a genuinely nonlocal set of $N$-partite orthogonal states. (Note that if a set of orthogonal states is locally indistinguishable for each bipartition of the subsystems, then this set is said to be genuinely nonlocal \cite{li2020local,Sumit2019Genuinely}). Thus the information cannot be fully accessed by the subordinates even if any $k$ subordinates are collusive, where $k<N$.

  In the above scheme of local hiding of information, we only use genuine quantum nonlocality, and it does not require strong quantum nonlocality. These subordinates may obtain part of the information. However, if the boss encode the information in a strongly nonlocal set of $N$-partite orthogonal states, and these subordinates can only perform OPLMs, then the information will not be accessible at all. By Theorem~\ref{thm:OES}, the new scheme of local hiding of information can be used for  $N$-partite systems, $N\geq 3$.

In the local hiding scheme by using strongly nonlocal sets in Fig.~\ref{Fig:hiding}, there are two ways to reveal the whole information by these subordinates. One way is to perform global measurements when all  $N$ subordinates are collusive, since any set of orthogonal pure states can be perfectly distinguished by using global measurements \cite{computation2010}.
 Another way is to use additional quantum resources such as entanglement \cite{cohen2008understanding}. It is known that if some proper $N$-partite genuinely entangled resource state assists all these $N$ subordinates, then the information could be fully accessed by these subordinates \cite{Sumit2019Genuinely,Halder2019Strong,2020Strong}.

Recently, $k$-uniform quantum information masking was proposed in  \cite{shi2021k}. When  information  is masked in a set of $N$-partite states $\{\ket{\psi_j}\}$, it is required that the reductions to $k$ parties of $\ket{\psi_j}$ are identical for all $j$. That is any $k$ parties have no information about the value $j$. There are several differences between $k$-uniform quantum information masking and the above local hiding of information. The former is based on reductions of states of a set, while the latter is based on local indistinguishability of a set.  In $k$-uniform quantum information, the masked information is completely inaccessible to each set of $k$ parties. However, in the local hiding scheme of information, part of information remains hidden if $k$ parties are collusive. In a sense, our scheme of local hiding of information is weaker than the scheme of $k$-uniform quantum information masking, which can also be seen from the range of parameters. In $k$-uniform quantum information masking, the parameters satisfy $k\leq\fl{N}{2}$ \cite{shi2021k} if the states from a whole space are masked. While in local hiding of information, $k$ has a wide range, that is $k\leq N-1$.

\section{Conclusion and discussion}\label{sec:con}
In this paper, we constructed a set of orthogonal entangled states in $(\bbC^d)^{\otimes N}$ for $d\geq 2$ and $N\geq 3$	from cyclic permutation group action, and showed strong quantum nonlocality in general $N$-partite systems for the first time. This result positively answered an open questions proposed in Refs.~\cite{2020Strong}. Further, when $N=3,4$ and $d\geq 3$, we presented a strongly nonlocal orthogonal genuinely entangled set, which has a smaller size than that of all known strongly nonlocal OPS in the same systems. Finally, we connected strong quantum nonlocality with local hiding of information.

There are also some open questions left. Whether we can show strong quantum nonlocality in $\bbC^{d_1}\otimes \bbC^{d_2}\otimes \cdots\otimes\bbC^{d_N}$ for any $d_i\geq 2$ and $1\leq i\leq N$?
How to construct a strongly nonlocal orthogonal genuinely entangled set in $(\bbC^d)^{\otimes N}$ for any $d\geq 2$ and $N\geq 5$? Whether we can use other permutation group actions to construct a strongly nonlocal orthogonal entangled set? Can we construct a strongly nonlocal UPB  for $N$-partite systems with $N\geq 5$? What is the minimum size of a strongly nonlocal OPS in $(\bbC^d)^{\otimes N}$?

\section*{Acknowledgments}
\label{sec:ack}		
We thank Mao-Sheng Li and Jun Gao for discussing this problem. F.S., Z.Y. and X.Z. were supported by the NSFC under Grants No. 11771419 and No. 12171452, the Anhui Initiative in Quantum Information Technologies under Grant No. AHY150200, and the National Key Research and Development Program of China 2020YFA0713100. L.C. was supported by the  NNSF of China (Grant No. 11871089), and the Fundamental Research Funds for the Central Universities (Grant No. ZG216S2005).

\vspace{5cm}
\onecolumngrid
\appendix

\section{The proof of Lemma~\ref{lem:B_d}}\label{appendix:lem:B_d}
	\begin{proof}	
	We have shown that 	$|\cB_d^N|=d^N-(d-1)^N+1$. Next, we show that $\cB_d^N$  is an OES.  Obviously,  $\cS_{(0,0,\ldots,0)}$ is an OES. We only need to consider $\cS_x$  for $x\neq (0,0\ldots,0)\in \bbX_d^N$. Assume  $\cO_{x}=\{(i_1^{(j)},i_2^{(j)},\ldots,i_N^{(j)})\mid j\in\bbZ_k\}$, $k\geq 2$. Since any two elements of $\cO_{x}=\{(i_1^{(j)},i_2^{(j)},\ldots,i_N^{(j)})\mid j\in\bbZ_k\}$ are distinct,  the states in $\{\ket{i_1^{(j)}}_{A_1}\ket{i_2^{(j)}}_{A_2}\cdots\ket{i_N^{(j)}}_{A_N}\mid j\in\bbZ_k\}$ are mutually orthogonal. Since the coefficient matrix
	$B=(w_k^{sj})_{s,j\in\bbZ_k}$ is a complex Hadamard matrix of order $k$, the states in $S_x$ must be mutually orthogonal. Next, we show that any state in $\cS_x$ is an entangled state.
	
	Without loss of generality, we assume that
	$x=(i_1^{(0)},i_2^{(0)},\ldots,i_N^{(0)})$.  By definition, $\cO_x=\{\sigma^rx\mid r\in\bbZ_N\}$. We have the following two claims.
	\begin{enumerate}[(i)]
		\item  We claim that the $(N-1)$-tuples $\{(i_2^{(j)},\ldots,i_N^{(j)})\mid j\in\bbZ_k\}$ must be mutually distinct. If  $(i_2^{(j)},\ldots,i_N^{(j)})=(i_2^{(j')},\ldots,i_N^{(j')})$ for some $j\neq j'\in\bbZ_k$, then $i_1^{(j)}=i_1^{(j')}$. Otherwise, if $i_1^{(j)}\neq i_1^{(j')}$, then   $(i_1^{(j)},i_2^{(j)},\ldots,i_N^{(j)})\notin \cO_x$ or $(i_1^{(j')},i_2^{(j')},\ldots,i_N^{(j')})\notin\cO_x$. We obtain that $(i_1^{(j)},i_2^{(j)},\ldots,i_N^{(j)})=(i_1^{(j')},i_2^{(j')},\ldots,i_N^{(j')})$, which is impossible. We obtain this claim.
		
		\item  We also claim that there must exist $j\neq j'\in\bbZ_k$, such that $i_1^{(j)}=0$ and $i_1^{(j')}\neq 0$. Since $x\neq (0,0\ldots,0)\in \bbX_d^N$, there must exist  $1\leq t\neq t'\leq N$, such that  $i_t^{(0)}=0$ and $i_{t'}^{(0)}\neq 0$. Further, there must exist $j\neq j'\in\bbZ_k$, such that $(i_1^{(j)},i_2^{(j)},\ldots,i_N^{(j)})=\sigma^{(t-1)}x$, and $(i_1^{(j')},i_2^{(j')},\ldots,i_N^{(j')})=\sigma^{(t'-1)}x$. Then $i_1^{(j)}=i_t^{(0)}=0$ and $i_1^{(j')}=i_{t'}^{(0)}\neq0$. We obtain this claim.
	\end{enumerate}	
	
	For $s\in\bbZ_k$, let
	\begin{equation*}
	\ket{\psi_s}=\sum_{j\in\bbZ_k}w_k^{sj}\ket{i_1^{(j)}}_{A_1}\ket{i_2^{(j)}}_{A_2} \cdots\ket{i_N^{(j)}}_{A_N}\in S_x.
	\end{equation*}
	By the above two claims, we know that $\rank(\ket{\psi_s}_{A_1|A_2\ldots A_N})\geq 2$. Thus, the set $\cS_x$ is an OES. Further, if $\cO_x\cap \cO_y=\emptyset$, then any state in $\cS_x$ is orthogonal to any state in $\cS_y$. Thus, $\cB_d^N$ is an OES.
\end{proof}
\vspace{0.4cm}

\section{Two lemmas for showing the strong nonlocality}\label{appendix:two_lemmas}
\begin{lemma}\label{lemma:zero}
	Assume that $\cB_i:=\{\ket{0},\ket{1},\ldots,\ket{d_i-1}\}$  is the computational basis of $\bbC^{d_i}$ for $i=1,2$. Let  a  $d_2\times d_2$ Hermitian matrix $E=(a_{i,j})_{i,j\in\bbZ_{d_2}}$ be the matrix representation of a Hermitian operator  $E$  under the basis  $\cB_2$.  Suppose that we have  two orthogonal sets in $\bbC^{d_1}\otimes \bbC^{d_2}$,
	\begin{equation}
		\begin{aligned}
			\{\ket{\psi_{i}}&=\sum_{j\in\bbZ_m}b_{i,j}\ket{p_j}\ket{r_j}\}_{i\in\bbZ_m},
			\\
			\{\ket{\varphi_{k}}&=\sum_{\ell\in\bbZ_n}c_{k,\ell}\ket{q_\ell}\ket{s_\ell}\}_{k\in\bbZ_n},
		\end{aligned}
	\end{equation}
	where  $\{\ket{r_j}\}_{j\in\bbZ_m}$, $\{\ket{s_\ell}\}_{\ell\in\bbZ_n}$ are two  nonempty  subsets of $\cB_2$, and $\ket{p_j}$, $\ket{q_\ell}\in \cB_1$ for $j\in \bbZ_m$ and $\ell\in\bbZ_n$. Here we do not require that these $p_j$'s ($q_l$'s, respectively) are distinct. Further, assume that
	\begin{equation}
		\bra{\psi_i}\bbI\otimes E\ket{\varphi_k}=0,  \ \text{for} \ i\in\bbZ_m, \ k\in\bbZ_n.
	\end{equation}
	If $p_j=q_\ell$ for some $j\in \bbZ_m$ and $\ell\in \bbZ_n$, then
	\begin{equation}
		a_{r_j,s_\ell}=a_{s_\ell,r_j}=0.
	\end{equation}
\end{lemma}
\begin{proof}
	Let $B=(b_{i,j})_{i,j\in\bbZ_m}$ and  $C=(c_{k,\ell})_{k,\ell\in\bbZ_n}$. Since $\{\ket{\psi_{i}}\}_{i\in\bbZ_m}$ and    $\{\ket{\varphi_{k}}\}_{k\in\bbZ_n}$ are two orthogonal sets, we obtain that $B$ and $C$ are both full-rank.  Denote $\overline{B}=(\overline{b_{i,j}})_{i,j\in\bbZ_m}$, where $\overline{b_{i,j}}$ is the complex conjugate of $b_{i,j}$.  Then  $\overline{B}\otimes C$ is also full-rank. Next,
	\begin{equation}\label{eq:orthogonal}
	\begin{aligned}
	\bra{\psi_i}\bbI\otimes E\ket{\varphi_k}&=\sum_{j\in\bbZ_m}\sum_{\ell\in\bbZ_n}\overline{b_{i,j}}c_{k,\ell}\braket{p_j}{q_\ell}\bra{r_j}E\ket{s_\ell}
	=\sum_{j\in\bbZ_m}\sum_{\ell\in\bbZ_n}\overline{b_{i,j}}c_{k,\ell}\braket{p_j}{q_\ell}a_{r_j,s_\ell}
	=\bu_{i,k}\cdot\bX=0,
	\end{aligned}
	\end{equation}
	where
	\begin{equation*}
	\begin{aligned}
	\bu_{i,k}&=(\overline{b_{i,0}}c_{k,0}, \ \overline{b_{i,0}}c_{k,1}, \ \cdots, \ \overline{b_{i,0}}c_{k,n-1}, \ \cdots, \ \overline{b_{i,m-1}}c_{k,0}, \ \overline{b_{i,m-1}}c_{k,1}, \ \cdots, \ \overline{b_{i,m-1}}c_{k,n-1})\\
	&=(\overline{b_{i,0}}, \ \overline{b_{i,1}}, \ \cdots, \ \overline{b_{i,m-1}})\otimes (c_{k,0}, \ c_{k,1}, \ \cdots, \ c_{k,n-1})
	\end{aligned}
	\end{equation*}
	and
	\begin{equation*}
	\begin{aligned}
	\bX=(&\braket{p_0}{q_0}a_{r_0,s_0}, \ \braket{p_0}{q_1}a_{r_0,s_1}, \ \cdots, \ \braket{p_0}{q_{n-1}}a_{r_0,s_{n-1}}, \ \cdots, \ \\
	&\braket{p_{m-1}}{q_0}a_{r_{m-1},s_0}, \ \braket{p_{m-1}}{q_1}a_{r_{m-1},s_1}, \ \cdots, \ \braket{p_{m-1}}{q_{n-1}}a_{r_{m-1},s_{n-1}})^{T}.
	\end{aligned}
	\end{equation*}
	Since Eq.~\eqref{eq:orthogonal} holds for any $i\in \bbZ_m$, and $k\in \bbZ_n$. We obtain that
	\begin{equation*}
	\overline{B}\otimes C\cdot \bX=\textbf{0}.
	\end{equation*}
	Since $\overline{B}\otimes C$ is full-rank, we have $\bX=\textbf{0}$.  That is
	\begin{equation*}
	\braket{p_j}{q_\ell}a_{r_j,s_\ell}=0, \ \text{for} \ j\in\bbZ_m, \ \ell\in\bbZ_n.
	\end{equation*}
	If $p_j=q_\ell$ for some $j\in \bbZ_m$ and $\ell\in \bbZ_n$, then we obtain that
	\begin{equation*}
	a_{r_j,s_\ell}=0.
	\end{equation*}
	Since $E=E^{\dagger}$, we also have $a_{s_\ell,r_j}=0$.	
\end{proof}	
\vspace{0.4cm}

\begin{lemma}\label{lemma:trivial}
	Assume that $\cB_i:=\{\ket{0},\ket{1},\ldots,\ket{d_i-1}\}$  is the computational basis of $\bbC^{d_i}$ for $i=1,2$. Let  a  $d_2\times d_2$ Hermitian matrix $E=(a_{i,j})_{i,j\in\bbZ_{d_2}}$ be the matrix representation of a Hermitian operator  $E$  under the basis  $\cB_2$.  Suppose that we have an orthogonal set in $\bbC^{d_1}\otimes \bbC^{d_2}$,
	\begin{equation}
		\begin{aligned}
			\{\ket{\psi_{i}}&=\sum_{j\in\bbZ_m}b_{i,j}\ket{p_j}\ket{r_j}\}_{i\in\bbZ_m},
		\end{aligned}
	\end{equation}
	where $\{\ket{r_j}\}_{j\in\bbZ_m}$ is a  subset  of $\cB_2$, and $\ket{p_j}\in \cB_1$ for $j\in \bbZ_m$.  Here we do not require that these $p_j$'s are distinct. Further, assume that
	\begin{equation}
		\bra{\psi_i}\bbI\otimes E\ket{\psi_k}=0,  \ \text{for} \ i\neq k\in\bbZ_m.
	\end{equation}
	If there exists  $t\in \bbZ_m$, such that
	$a_{r_t,r_j}=0$ for $j\neq t\in\bbZ_m$, and $b_{i,t}\neq 0$ for $i\in\bbZ_m$, then we obtain that
	\begin{equation}
		a_{r_i,r_j}=a_{r_j,r_i}=0,  \ \text{for} \ i\neq j\ \text{and} \ p_i=p_j,
	\end{equation}
	and
	\begin{equation}
		a_{r_0,r_0}=a_{r_i,r_i},   \ \text{for} \ i\neq 0\in\bbZ_m.
	\end{equation}
\end{lemma}
\begin{proof}
	Without loss of generality, we can assume that $t=0$. Let $B=(b_{i,j})_{i,j\in\bbZ_m}$,  $C=(c_{i,j})_{i,j\in\bbZ_m}$, where $c_{i,j}=b_{i,j}/\sqrt{\sum_{j\in\bbZ_m}|b_{i,j}|^2}$. First, we can normalize the states 	$\{\ket{\psi_{i}}=\sum_{j\in\bbZ_m}b_{i,j}\ket{p_j}\ket{r_j}\}_{i\in\bbZ_m}$ as $\{\ket{\varphi_{i}}=\sum_{j\in\bbZ_m}c_{i,j}\ket{p_j}\ket{r_j}\}_{i\in\bbZ_m}$. Since the states $\{\ket{\psi_{i}}\}_{i\in\bbZ_m}$ are mutually orthogonal, we obtain that the row vectors of $B$ are mutually orthogonal. This implies that $C$ is a unitary matrix. We can also obtain that
	\begin{equation*}
	\bra{\varphi_i}\bbI\otimes E\ket{\varphi_k}=0,  \ \text{for} \ i\neq k\in\bbZ_m.
	\end{equation*}
   By the same discussion as Eq.~\eqref{eq:orthogonal}, we have
	\begin{equation}\label{eq:orthogonal_1}
	\bra{\varphi_i}\bbI\otimes E\ket{\varphi_k}=\bu_{i,k}\cdot \textbf{X}=0,
	\end{equation}
	where
	\begin{equation*}
	\bu_{i,k}=(\overline{c_{i,0}}, \ \overline{c_{i,1}}, \ \cdots, \ \overline{c_{i,m-1}})\otimes (c_{k,0}, \ c_{k,1}, \ \cdots, \ c_{k,m-1}),
	\end{equation*}
	and
	\begin{equation*}
	\begin{aligned}
	\bX=(&\braket{p_0}{p_0}a_{r_0,r_0}, \ \braket{p_0}{p_1}a_{r_0,r_1}, \ \cdots, \ \braket{p_0}{p_{m-1}}a_{r_0,r_{m-1}}, \ \cdots, \ \\
	&\braket{p_{m-1}}{p_0}a_{r_{m-1},r_0}, \ \braket{p_{m-1}}{p_1}a_{r_{m-1},r_1}, \ \cdots, \ \braket{p_{m-1}}{p_{m-1}}a_{r_{m-1},r_{m-1}})^{T}.
	\end{aligned}
	\end{equation*}
	Note that
		\begin{equation*}
	\overline{C}\otimes C=
	\left(\bu_{0,0}^T,
	\bu_{0,1}^T,
	\ldots\\
	\bu_{0,m-1}^T,
	\ldots\\
	\bu_{m-1,0}^T,
	\bu_{m-1,1}^T,
	\ldots\\
	\bu_{m-1,m-1}^T\right)^T
	\end{equation*}
	is still a unitary matrix. Since Eq.~\eqref{eq:orthogonal_1} holds for any
	$i\neq k$, there exists $e_i\in \bbC$ for each $i\in\bbZ_m$ such that
	\begin{equation}\label{eq:X}
	\bX=e_0\bu_{0,0}^{\dagger}+e_1\bu_{1,1}^{\dagger}+\cdots+e_{m-1}\bu_{m-1,m-1}^{\dagger}.
	\end{equation}

	By the condition, we have $a_{r_0,r_j}=0$ for $j\neq 0\in\bbZ_m$, and $b_{i,0}\neq 0$ for $i\in\bbZ_m$. This also implies that $c_{i,0}\neq 0$ for $i\in\bbZ_m$. By Eq.~\eqref{eq:X}, we have
	\begin{equation}\label{eq:expand_0}
	\begin{pmatrix}
	0\\
	\vdots\\
	0\\
	\end{pmatrix}=
	\begin{pmatrix}
	\braket{p_0}{p_1}a_{r_0,r_1}\\
	\vdots\\
	\braket{p_0}{p_{m-1}}a_{r_0,r_{m-1}}\\
	\end{pmatrix}=
	\begin{pmatrix}
	\overline{c_{0,1}} & \overline{c_{1,1}} &\cdots & \overline{c_{m-1,1}}\\
	\vdots &\vdots &\ddots &\vdots\\
	\overline{c_{0,m-1}} & \overline{c_{1,m-1}}  &\cdots &\overline{c_{m-1,m-1}}\\
	\end{pmatrix}
	\begin{pmatrix}
	c_{0,0}e_0\\
	\vdots\\
	c_{m-1,0}e_{m-1}
	\end{pmatrix}.
	\end{equation}
	Since $C^{\dagger}$ is a unitary matrix,  there exists  $k\in \bbC$,
	such that $c_{i,0}e_i=kc_{i,0}$ for $i\in \bbZ_m$.
	Further, since  $c_{i,0}\neq 0$ for $i\in\bbZ_m$, we have 	$e_i=k, \ \text{for} \ i\in\bbZ_m$.
	Then by Eq.~\eqref{eq:X}, we obtain that
	\begin{equation*}
	\braket{p_i}{p_j}a_{r_i,r_j}=\delta_{i,j}k, \ \text{for} \ i,j\in\bbZ_m.
	\end{equation*}
	Thus, if $i\neq j$, and $p_i=p_j$, then we have
	\begin{equation*}
	a_{r_i,r_j}=0.
	\end{equation*}
	Since $E=E^{\dagger}$, we also have $a_{r_j,r_i}=0$.	
	If $i=j$, then we have
	\begin{equation*}
	a_{r_i,r_i}=k,
	\end{equation*}
	i.e.  $a_{r_0,r_0}=a_{r_i,r_i},   \ \text{for} \ j\neq 0\in\bbZ_m.$
\end{proof}
\vspace{0.4cm}		

\section{The proof of Theorem~\ref{thm:OES}}\label{appendix:thm}	
	\begin{proof}
		Let $A_2,A_3,\ldots, A_{N}$ come together to  perform a joint  OPLM $\{E=M^{\dagger}M\}$, where each POVM element can be  written as $E=M^{\dagger}M=(a_{i_1i_2\cdots i_{N-1},j_1j_2\cdots j_{N-1}})_{i_k,j_k\in\bbZ_d, 1\leq k\leq N-1}$ under the basis $\{\ket{i_1}_{A_2}\ket{i_2}_{A_3}\cdots\ket{i_{N-1}}_{A_N}\}_{i_k\in\bbZ_d, 1\leq k\leq N-1}$.  Then the postmeasurement states $\{\bbI_{A_1}\otimes M\ket{\psi}: \ket{\psi}\in \cB_d^N\}$ should be mutually orthogonal.
		
		For an $(i_1,i_2,\ldots,i_{N-1})\in\bbZ_d^{N-1}$,  we denote  $wt(i_1,i_2,\ldots,i_{N-1})$ as the number of nonzero $i_k$ for $1\leq k\leq N-1$. We also define $N$ subsets of $\bbZ_d^{N-1}$,
		\begin{equation}
		\cA_k:=\{(i_1,i_2,\ldots,i_{N-1})\in \bbZ_d^{N-1}\mid wt(i_1,i_2,\ldots,i_{N-1})=k\}, \ k\in\bbZ_N.\\
		\end{equation}
		Note that $\cA_k\cap\cA_\ell=\emptyset$ for $k\neq \ell\in\bbZ_N$, and  $\bbZ_d^{N-1}=\mathop{\bigcup}\limits_{k\in\bbZ_N}\cA_k$. First, we need to show that the off-diagonal elements of $E$ are all zeros. That is
		$a_{i_1i_2\cdots i_{N-1},j_1j_2\cdots j_{N-1}}=0$ for $(i_1,i_2,\ldots, i_{N-1})\neq (j_1,j_2,\ldots, j_{N-1})\in\bbZ_d^{N-1}$. There are two cases.
		
		\begin{enumerate}[(i)]
			\item Assume $(i_1,i_2,\ldots,i_{N-1})\in \cA_k$, $(j_1,j_2,\ldots,j_{N-1})\in\cA_\ell$ for $k\neq \ell\in\bbZ_N$. Then we must have  $\cO_{(0,i_1,i_2,\ldots,i_{N-1})}\cap\cO_{(0,j_1,j_2,\ldots,j_{N-1})}=\emptyset$, and $\cS_{(0,i_1,i_2,\ldots,i_{N-1})}\neq \cS_{(0,i_1,i_2,\ldots,i_{N-1})}$.   Applying Lemma~\ref{lemma:zero} to $\cS_{(0,i_1,i_2,\ldots,i_{N-1})}$ and  $\cS_{(0,i_1,i_2,\ldots,i_{N-1})}$, we obtain that $a_{i_1i_2\cdots i_{N-1},j_1j_2\cdots j_{N-1}}=0$.
			
			\item Assume $(i_1,i_2,\ldots,i_{N-1})$, $(j_1,j_2,\ldots,j_{N-1})\in\cA_k$ for $k\neq 0\in\bbZ_N$.
			\begin{enumerate}[(1)]	
				\item 	First, we consider the case $k=1$. Then there must exist an $\ell$ such that $i_{\ell}\neq 0$ and $i_{m}=0$ for any $1\leq m\leq N-1$ and $m\neq \ell$.  If $\cO_{(0,i_1,i_2,\ldots,i_{N-1})}\cap\cO_{(0,j_1,j_2,\ldots,j_{N-1})}=\emptyset$, then we can also obtain that $a_{i_1i_2\cdots i_{N-1},j_1j_2\cdots j_{N-1}}=0$ for the same discussion as (i). If $\cO_{(0,i_1,i_2,\ldots,i_{N-1})}=\cO_{(0,j_1,j_2,\ldots,j_{N-1})}$, then $\sigma^{\ell}(0,i_1,i_2,\ldots,i_{N-1})=(i_{\ell},0,0,\ldots,0)\in\cO_{(0,i_1,i_2,\ldots,i_{N-1})}$, where $(0,0,\ldots,0)\in \cA_0$. For any $(n_0,n_1,\ldots,n_{N-1})\in \cO_{(0,i_1,i_2,\ldots,i_{N-1})}$, we must have $(n_1,\ldots,n_{N-1})\in \cA_0$ or $\cA_1$. For $(i_1',i_2',\ldots, i_{N-1}')\neq (j_1',j_2',\ldots, j_{N-1}')\in\bbZ_d^{N-1}$,    we have shown that $a_{i_1'i_2'\cdots i_{N-1}',j_1'j_2'\cdots j_{N-1}'}=0$ for  $(i_1',i_2',\ldots,i_{N-1}')\in \cA_0$, $(j_1',j_2',\ldots,j_{N-1}')\in\cA_1$. This implies that $a_{00\cdots 0,n_1n_2\cdots n_{N-1}}=0$ for any  $(n_0,n_1,\ldots,n_{N-1})\neq (i_\ell,0,\ldots,0)\in \cO_{(0,i_1,i_2,\ldots,i_{N-1})}$.    Then we can apply Lemma~\ref{lemma:trivial} to $\cS_{(0,i_1,i_2,\ldots,i_{N-1})}$, and we obtain that $a_{i_1i_2\cdots i_{N-1},j_1j_2\cdots j_{N-1}}=0$  for $(i_1,i_2,\ldots,i_{N-1})$, $(j_1,j_2,\ldots,j_{N-1})\in\cA_1$.
				
				\item	 Next, we consider the case $k=2$.  Then there must exist an  $i_\ell\neq 0$, where  $1\leq \ell\leq N-1$. We only need to consider the case $\cO_{(0,i_1,i_2,\ldots,i_{N-1})}=\cO_{(0,j_1,j_2,\ldots,j_{N-1})}$. We denote $i_0=0$. Then $\sigma^{\ell}(0,i_1,i_2,\ldots,i_{N-1})=(i_\ell,i_{\ell+1},i_{\ell+2},\ldots,i_{\ell-1})\in \cO_{(0,i_1,i_2,\ldots,i_{N-1})}$, where $(i_{\ell+1},i_{\ell+2},\ldots,i_{\ell-1})\in \cA_1$.   For any $(n_0,n_1,\ldots,n_{N-1})\in \cO_{(0,i_1,i_2,\ldots,i_{N-1})}$, we must have $(n_1,\ldots,n_{N-1})\in \cA_1$ or $\cA_2$.  For $(i_1',i_2',\ldots, i_{N-1}')\neq (j_1',j_2',\ldots, j_{N-1}')\in\bbZ_d^{N-1}$,    we have shown that $a_{i_1'i_2'\cdots i_{N-1}',j_1'j_2'\cdots j_{N-1}'}=0$ for  $(i_1',i_2',\ldots,i_{N-1}')\in \cA_1$, $(j_1',j_2',\ldots,j_{N-1}')\in\cA_2$, and $a_{i_1'i_2'\cdots i_{N-1}',j_1'j_2'\cdots j_{N-1}'}=0$ for  $(i_1',i_2',\ldots,i_{N-1}')$, $(j_1',j_2',\ldots,j_{N-1}')\in\cA_1$. This implies that  $a_{i_{\ell+1},i_{\ell+2}\cdots,i_{\ell-1},n_1n_2\cdots n_{N-1}}=0$ for any  $(n_0,n_1,\ldots,n_{N-1})\neq (i_\ell,i_{\ell+1},\ldots,i_{\ell-1})\in \cO_{(0,i_1,i_2,\ldots,i_{N-1})}$.  Then we can apply Lemma~\ref{lemma:trivial} to $\cS_{(0,i_1,i_2,\ldots,i_{N-1})}$, we obtain that $a_{i_1i_2\cdots i_{N-1},j_1j_2\cdots j_{N-1}}=0$  for $(i_1,i_2,\ldots,i_{N-1})$, $(j_1,j_2,\ldots,j_{N-1})\in\cA_2$.
				
				\item	 By repeating this process $N-2$ times, we obtain that $a_{i_1i_2\cdots i_{N-1},j_1j_2\cdots j_{N-1}}=0$,  for $(i_1,i_2,\ldots,i_{N-1})$, $(j_1,j_2,\ldots,j_{N-1})\in\cA_k$ for $1\leq k\leq N-2$.
				
				\item	 Finally, we consider the case $k=N-1$. If $(i_1,i_2,\ldots,i_{N-1})\neq (j_1,j_2,\ldots,j_{N-1})\in\cA_{N-1}$, then
				$\cO_{(0,i_1,i_2,\ldots,i_{N-1})}\cap\cO_{(0,j_1,j_2,\ldots,j_{N-1})}=\emptyset$. By the same discussion as (i),  we obtain that $a_{i_1i_2\cdots i_{N-1},j_1j_2\cdots j_{N-1}}=0$,  for $(i_1,i_2,\ldots,i_{N-1})$, $(j_1,j_2,\ldots,j_{N-1})\in\cA_{N-1}$.
			\end{enumerate}
		\end{enumerate}
		
		In a word, we know that the off-diagonal elements of $E$ are all zeros. Next, we consider the diagonal elements of $E$.
		
		For any $(i_1,i_2,\ldots,i_{N-1})\in \cA_k$ for $k\neq 0\in\bbZ_N$, there must exist an $(n_{(0,k-1)},n_{(1,k-1)},\ldots,n_{(N-1,k-1)})\in \cO_{(0,i_1,i_2,\ldots,i_{N-1})}$, such that $(n_{(1,k-1)},\ldots,n_{(N-1,k-1)})\in\cA_{k-1}$. Applying Lemma~\ref{lemma:trivial} to $S_{(0,i_1,i_2,\ldots,i_{N-1})}$, we obtain that $a_{i_1i_2\cdots i_{N-1},i_1i_2\cdots i_{N-1}}=a_{n_{(1,k-1)}n_{(2,k-1)}\cdots n_{(N-1,k-1)}, n_{(1,k-1)}n_{(2,k-1)}\cdots n_{(N-1,k-1)}}$. Next, there  must exist an $(n_{(0,k-2)},n_{(1,k-2)},\ldots,n_{(N-1,k-2)})\in \cO_{(0,n_{(1,k-1)},\ldots,n_{(N-1,k-1)})}$, such that $(n_{(1,k-2)},\ldots,n_{(N-1,k-2)})\in\cA_{k-2}$. Applying Lemma~\ref{lemma:trivial} to $ S_{(0,n_{(1,k-1)},\ldots,n_{(N-1,k-1)})}$, we obtain that $a_{n_{(1,k-1)}n_{(2,k-1)}\cdots n_{(N-1,k-1)}, n_{(1,k-1)}n_{(2,k-1)}\cdots n_{(N-1,k-1)}}=a_{n_{(1,k-2)}n_{(2,k-2)}\cdots n_{(N-1,k-2)}, n_{(1,k-2)}n_{(2,k-2)}\cdots n_{(N-1,k-2)}}$. By repeating this process $k$ times, we obtain that  $a_{i_1i_2\cdots i_{N-1},i_1i_2\cdots i_{N-1}}=a_{n_{(1,0)}n_{(2,0)}\cdots n_{(N-1,0)}, n_{(1,0)}n_{(2,0)}\cdots n_{(N-1,0)}}$, where $(n_{(1,0)},n_{(2,0)},\cdots, n_{(N-1,0)})\in\cA_0$. That is   $a_{i_1i_2\cdots i_{N-1},i_1i_2\cdots i_{N-1}}=a_{00\cdots 0, 00\cdots 0}$. We obtain that the diagonal elements of $E$ are all equal.  Therefore, $E$ is trivial. This completes the proof.
	\end{proof}

\section{The proof of Lemma~\ref{pro:OGES_dddd}}\label{appendix:OGES_dddd}
 Since

\begin{equation*}
\begin{aligned}
\bbX_d^4=&\cO_{(0,0,0,0)}\bigcup\left(\mathop{\cup}\limits_{i\in[d-1]}(\cO_{(0,0,0,i)}\cup\cO_{(0,0,i,i)}\cup\cO_{(0,i,0,i)})\right)\bigcup \left(\mathop{\cup}\limits_{ i\neq  j\in [d-1]}\cO_{(0,0,i,j)}\right)
\bigcup\left(\mathop{\cup}\limits_{p<q\in [d-1]}\cO_{(0,p,0,q)}\right)\\&\bigcup\left(\mathop{\cup}\limits_{k,\ell,m\in [d-1]}\cO_{(0,k,\ell,m)}\right),
\end{aligned}
\end{equation*}
we obtain that
\begin{equation*}
\begin{aligned}	
\overline{\cB_d^4}=&\cS_{(0,0,0,0)}\bigcup\left(\mathop{\cup}\limits_{i\in[d-1]}(\cS_{(0,0,0,i)}\cup\overline{\cS_{(0,0,i,i)}}\cup\cS_{(0,i,0,i)})\right)\bigcup \left(\mathop{\cup}\limits_{ i\neq  j\in [d-1]}\cS_{(0,0,i,j)}\right)\bigcup\left(\mathop{\cup}\limits_{p<q\in [d-1]}\cS_{(0,p,0,q)}\right)\\&\bigcup\left(\mathop{\cup}\limits_{k,\ell,m\in [d-1]}\cS_{(0,k,\ell,m)}\right),
\end{aligned}
\end{equation*}
where

\begin{equation*}
\begin{aligned}
\cS_{(0,0,0,0)}=&\{\ket{0}_{A_1}\ket{0}_{A_2}\ket{0}_{A_3}\ket{0}_{A_4}\pm\ket{1}_{A_1}\ket{1}_{A_2}\ket{1}_{A_3}\ket{1}_{A_4}\},\\
\cS_{(0,0,0,i)}=&\{\ket{0}_{A_1}\ket{0}_{A_2}\ket{0}_{A_3}\ket{i}_{A_4}+w_4^{s}\ket{0}_{A_1}\ket{0}_{A_2}\ket{i}_{A_3}\ket{0}_{A_4}+w_4^{2s}\ket{0}_{A_1}\ket{i}_{A_2}\ket{0}_{A_3}\ket{0}_{A_4}+\\&w_4^{3s}\ket{i}_{A_1}\ket{0}_{A_2}\ket{0}_{A_3}\ket{0}_{A_4}: s\in\bbZ_4\},\\
\overline{\cS_{(0,0,i,i)}}=&\{b_{s,0}\ket{0}_{A_1}\ket{0}_{A_2}\ket{i}_{A_3}\ket{i}_{A_4}+b_{s,1}\ket{0}_{A_1}\ket{i}_{A_2}\ket{i}_{A_3}\ket{0}_{A_4}+b_{s,2}\ket{i}_{A_1}\ket{i}_{A_2}\ket{0}_{A_3}\ket{0}_{A_4}\\&+b_{s,3}\ket{i}_{A_1}\ket{0}_{A_2}\ket{0}_{A_3}\ket{i}_{A_4}: s\in\bbZ_4\},\\
\cS_{(0,i,0,i)}=&\{\ket{0}_{A_1}\ket{i}_{A_2}\ket{0}_{A_3}\ket{i}_{A_4}\pm\ket{i}_{A_1}\ket{0}_{A_2}\ket{i}_{A_3}\ket{0}_{A_4}\},\\
\cS_{(0,0,i,j)}=&\{\ket{0}_{A_1}\ket{0}_{A_2}\ket{i}_{A_3}\ket{j}_{A_4}+w_4^s\ket{0}_{A_1}\ket{i}_{A_2}\ket{j}_{A_3}\ket{0}_{A_4}+w_4^{2s}\ket{i}_{A_1}\ket{j}_{A_2}\ket{0}_{A_3}\ket{0}_{A_4}\\&+w_4^{3s}\ket{j}_{A_1}\ket{0}_{A_2}\ket{0}_{A_3}\ket{i}_{A_4}: s\in \bbZ_4\},\\
\cS_{(0,p,0,q)}=&\{\ket{0}_{A_1}\ket{p}_{A_2}\ket{0}_{A_3}\ket{q}_{A_4}+w_4^s\ket{p}_{A_1}\ket{0}_{A_2}\ket{q}_{A_3}\ket{0}_{A_4}+w_4^{2s}\ket{0}_{A_1}\ket{q}_{A_2}\ket{0}_{A_3}\ket{p}_{A_4}\\&+w_4^{3s}\ket{q}_{A_1}\ket{0}_{A_2}\ket{p}_{A_3}\ket{0}_{A_4}: s\in\bbZ_4\},\\
\cS_{(0,k,\ell,m)}=&\{\ket{0}_{A_1}\ket{k}_{A_2}\ket{\ell}_{A_3}\ket{m}_{A_4}+w_4^s\ket{k}_{A_1}\ket{\ell}_{A_2}\ket{m}_{A_3}\ket{0}_{A_4}+w_4^{2s}\ket{\ell}_{A_1}\ket{m}_{A_2}\ket{0}_{A_3}\ket{k}_{A_4}\\&+w_4^{3s}\ket{m}_{A_1}\ket{0}_{A_2}\ket{k}_{A_3}\ket{\ell}_{A_4}: s\in\bbZ_4\}.\\
\end{aligned}
\end{equation*}
for  $i\neq j\in[d-1]$, $p<q\in[d-1]$ and $k,\ell,m\in[d-1]$. Note that
	\begin{equation*}
\widetilde{B}=(b_{i,j})_{i,j\in\bbZ_4}:=\begin{pmatrix}
1 &1 &1 &2\\
1 &-1 &2 &-1\\
5 &5 &-2 &-4\\
5 &-5 &-4 &2\\
\end{pmatrix}.
\end{equation*}
\vspace{0.4cm}

			\begin{proof}
		Since $|\cB_d^4|=d^4-(d-1)^4+1$,  and $|\overline{\cS_{(0,0,i,i)}}|=|\cS_{(0,0,i,i)}|$ for $i\in[d-1]$, we obtain that $|\overline{\cB_d^4}|=|\cB_d^4|=d^4-(d-1)^4+1$.  By Theorem~\ref{thm:OES}, we know that $\cB_d^4$ is strongly nonlocal. Since $\overline{\cB_d^4}$ has a similar structure to $\cB_d^4$. We can also obtain that  $\overline{\cB_d^4}$ is strongly nonlocal. We only need to show that  $\overline{\cB_d^4}$ is an OGES.  It is easy to see that any state in $\cS_{(0,0,0,0)}$ and $\cS_{(0,i,0,i)}$ is LU-equivalent to a $\ghz$ state for $i\in [d-1]$, and any state in $\cS_{(0,0,0,i)}$ is LU-equivalent to a W state for $i\in [d-1]$. For the set $\cS_{(0,0,i,j)}$, $\cS_{(0,p,0,q)}$ and $\cS_{(0,k,\ell,m)}$, we need to use the following claim.
		
		\noindent	\textbf{Claim}: For a  state $\ket{\psi}\in(\bbC^d)^{\otimes N}$, if there exists a product operator $P_1\otimes P_2\otimes \cdots \otimes P_N$ such that
		\begin{equation*}
		P_1\otimes P_2\otimes \cdots\otimes P_N \ket{\psi}
		\end{equation*}
		is a genuinely entangled state, then $\ket{\psi}$ is also a  genuinely entangled state.

		The proof of the above claim is as follows. If $\ket{\psi}$ is not a genuinely entangled state, then there exists a bipartition
		$A|B$ such that $\ket{\psi}=\ket{\psi}_A\otimes\ket{\psi}_B $. For any product operator $P_1\otimes P_2\otimes \cdots \otimes P_N$,
		\begin{equation*}
		P_1\otimes P_2\otimes \cdots\otimes P_N (\ket{\psi}_A\otimes\ket{\psi}_B)
		\end{equation*} is still separable in $A|B$ bipartition, which is not a genuinely entangled state. Contradiction. This completes the proof of this claim.
		
		For  $s\in \bbZ_4$, and $i\neq j\in[d-1]$, let
		\begin{equation*}
		\ket{\psi_s}=\ket{0}_{A_1}\ket{0}_{A_2}\ket{i}_{A_3}\ket{j}_{A_4}+w_4^s\ket{0}_{A_1}\ket{i}_{A_2}\ket{j}_{A_3}\ket{0}_{A_4}+w_4^{2s}\ket{i}_{A_1}\ket{j}_{A_2}\ket{0}_{A_3}\ket{0}_{A_4}+w_4^{3s}\ket{j}_{A_1}\ket{0}_{A_2}\ket{0}_{A_3}\ket{i}_{A_4}\in\cS_{(0,0,i,j)},
		\end{equation*}	
		and
		\begin{equation*}
		P=(\ketbra{0}{0}+\ketbra{0}{j}+w_4^{-2s}\ketbra{1}{i})_{A_1}(\ketbra{0}{0}+\ketbra{0}{j}+w_4^{-s}\ketbra{1}{i})_{A_2}	(\ketbra{0}{0}+\ketbra{0}{j}+\ketbra{1}{i})_{A_3}(\ketbra{0}{0}+\ketbra{0}{j}+w_4^{-3s}\ketbra{1}{i})_{A_4}.
		\end{equation*}
		Then
		\begin{equation*}
		P\ket{\psi_s}=\ket{0}_{A_1}\ket{0}_{A_2}\ket{1}_{A_3}\ket{0}_{A_4}+\ket{0}_{A_1}\ket{1}_{A_2}\ket{0}_{A_3}\ket{0}_{A_4}+\ket{1}_{A_1}\ket{0}_{A_2}\ket{0}_{A_3}\ket{0}_{A_4}+\ket{0}_{A_1}\ket{0}_{A_2}\ket{0}_{A_3}\ket{1}_{A_4}=\ket{W}_2^4,
		\end{equation*}
		which is a W-state. By the above claim,  $\ket{\psi_s}$ is a genuinely entangled state.  In the same way, we can also show that any state in $\cS_{(0,p,0,q)}$ and  $\cS_{(0,k,\ell,m)}$ is a genuinely entangled state for $p<q\in[d-1]$ and $k,\ell,m\in[d-1]$.

		Finally, we consider $\overline{\cS_{(0,0,i,i)}}$. For $s\in \bbZ_4$ and $i\in[d-1]$, let
		\begin{equation*}
		\begin{aligned}
		\ket{\lambda_s}=&b_{s,0}\ket{0}_{A_1}\ket{0}_{A_2}\ket{i}_{A_3}\ket{i}_{A_4}+b_{s,1}\ket{0}_{A_1}\ket{i}_{A_2}\ket{i}_{A_3}\ket{0}_{A_4}+b_{s,2}\ket{i}_{A_1}\ket{i}_{A_2}\ket{0}_{A_3}\ket{0}_{A_4}+b_{s,3}\ket{i}_{A_1}\ket{0}_{A_2}\ket{0}_{A_3}\ket{i}_{A_4}\in\overline{\cS_{(0,0,i,i)}}.
		\end{aligned}
		\end{equation*}
		We can calculate that $\rank(\ket{\lambda_s}_{A_1|A_2A_3A_4})=\rank(\ket{\lambda_s}_{A_2|A_3A_4A_1})=\rank(\ket{\lambda_s}_{A_3|A_4A_1A_2})=\rank(\ket{\lambda_s}_{A_4|A_1A_2A_3})=\rank(\ket{\lambda_s}_{A_1A_3|A_2A_4})= 2$, and
		$\rank(\ket{\lambda}_{A_1A_2|A_3A_4})=\rank(\ket{\lambda}_{A_1A_4|A_2A_3})=4$. This implies that $\ket{\lambda_s}$ is entangled in every bipartition, and it is a genuinely entangled state.
		
		In a word,  $\overline{B_d^4}$ is an OGES. This completes the proof.		
	\end{proof}
\vspace{0.4cm}	

		\begin{table*}[t]
	\renewcommand\arraystretch{1.7}	
	\caption{Off-diagonal elements and Diagonal elements of $E=M^{\dagger}M=(a_{ij,k\ell})_{i,j,k,\ell\in\bbZ_3}$. We apply Lemma~\ref{lemma:zero}  to two sets among $\{\cA_i\}_{i=1}^6$,  and apply Lemma~\ref{lemma:trivial}  to one set among $\{\cA_i\}_{i=1}^6$. }\label{Table:333}
	\centering
	\renewcommand\tabcolsep{5pt}
	\begin{tabular}{ll|ll}
		\midrule[1.1pt]
		Sets &Elements 	&Sets &Elements \\
		\midrule[1.1pt]
		$\cA_1$, $\cA_2$	  &$a_{00,02}=a_{00,20}=a_{00,22}=0$ &$\cA_2$, $\cA_5$   &$a_{02,21}=a_{20,21}=a_{00,10}=0$\\
		$\cA_1$, $\cA_3$   &$a_{00,11}=a_{10,11}=a_{01,11}=0$   &$\cA_3$, $\cA_4$   &$a_{11,12}=a_{10,20}=a_{01,20}=0$\\
		$\cA_1$, $\cA_4$   &$a_{00,12}=a_{11,20}=a_{01,22}=0$   &$\cA_3$, $\cA_5$   &$a_{11,21}=a_{02,10}=a_{01,02}=0$\\
		$\cA_1$, $\cA_5$   &$a_{00,21}=a_{02,11}=a_{10,22}=0$  	&$\cA_4$, $\cA_5$   &$a_{12,21}=a_{02,20}=a_{01,10}=0$\\
		$\cA_1$, $\cA_6$   &$a_{11,22}=a_{21,22}=a_{12,22}=0$ 	&$\cA_4$, $\cA_6$   &$a_{20,22}=a_{01,21}=a_{01,12}=0$\\
		$\cA_2$, $\cA_4$   &$a_{02,12}=a_{12,20}=a_{00,01}=0$ 	&$\cA_5$, $\cA_6$   &$a_{02,22}=a_{10,21}=a_{10,12}=0$\\
		\midrule[1.1pt]
		$\cA_1$   &$a_{00,00}=a_{11,11}=a_{22,22}$  &$\cA_4$   &$a_{12,12}=a_{20,20}$\\
		$\cA_2$   &$a_{02,02}=a_{20,20}=a_{00,00}$ &	$\cA_5$   &$a_{02,02}=a_{21,21}$\\
		$\cA_3$   &$a_{11,11}=a_{10,10}=a_{01,01}$\\
		\midrule[1.1pt]
	\end{tabular}
\end{table*}

\section{A strongly nonlocal OGES of size $18$ in $\bbC^3\otimes\bbC^3\otimes\bbC^3$}\label{appendix:OGES_333}	
We denote	
		\begin{equation}
	\begin{aligned}
	\cA_1&:=\{\ket{0}_{A_1}\ket{0}_{A_2}\ket{0}_{A_3}+w_3^s\ket{1}_{A_1}\ket{1}_{A_2}\ket{1}_{A_3}+w_3^{2s}\ket{2}_{A_1}\ket{2}_{A_2}\ket{2}_{A_3}\mid s\in\bbZ_3\},\\
	\cA_2&:=\{\ket{0}_{A_1}\ket{0}_{A_2}\ket{2}_{A_3}+w_3^s\ket{0}_{A_1}\ket{2}_{A_2}\ket{0}_{A_3}+w_3^{2s}\ket{2}_{A_1}\ket{0}_{A_2}\ket{0}_{A_3}\mid s\in\bbZ_3\},\\
	\cA_3&:=\{\ket{0}_{A_1}\ket{1}_{A_2}\ket{1}_{A_3}+w_3^s\ket{1}_{A_1}\ket{1}_{A_2}\ket{0}_{A_3}+w_3^{2s}\ket{1}_{A_1}\ket{0}_{A_2}\ket{1}_{A_3}\mid s\in\bbZ_3\},\\
	\cA_4&:=\{\ket{0}_{A_1}\ket{1}_{A_2}\ket{2}_{A_3}+w_3^s\ket{1}_{A_1}\ket{2}_{A_2}\ket{0}_{A_3}+w_3^{2s}\ket{2}_{A_1}\ket{0}_{A_2}\ket{1}_{A_3}\mid s\in\bbZ_3\},\\
	\cA_5&:=\{\ket{1}_{A_1}\ket{0}_{A_2}\ket{2}_{A_3}+w_3^s\ket{0}_{A_1}\ket{2}_{A_2}\ket{1}_{A_3}+w_3^{2s}\ket{2}_{A_1}\ket{1}_{A_2}\ket{0}_{A_3}\mid s\in\bbZ_3\},\\
	\cA_6&:=\{\ket{1}_{A_1}\ket{2}_{A_2}\ket{2}_{A_3}+w_3^s\ket{2}_{A_1}\ket{2}_{A_2}\ket{1}_{A_3}+w_3^{2s}\ket{2}_{A_1}\ket{1}_{A_2}\ket{2}_{A_3}\mid s\in\bbZ_3\}.\\
	\end{aligned}
	\end{equation}
	\vspace{0.4cm}
Then $\mathop{\cup}\limits_{1\leq i\leq 6}\cA_i$ is a strongly nonlocal OGES. The proof  is as follows. It is easy to see that any state in $\cA_1$, $\cA_4$ and $\cA_5$ is equivalent to a $\ghz$ state, and any state in $\cA_2$, $\cA_3$ and $\cA_6$	is equivalent to a W state. Thus $\mathop{\cup}\limits_{1\leq i\leq 6}\cA_i$ is an OGES. Next, we show the strong nonlocality.

  	Let $A_2$ and $A_3$ come together to  perform a joint  OPLM $\{E=M^{\dagger}M\}$, where each POVM element can be  written as $E=M^{\dagger}M=(a_{ij,k\ell})_{i,j,k,\ell\in\bbZ_3}$ under the basis $\{\ket{i}_{A_2}\ket{j}_{A_3}\}_{i,j\in\bbZ_3}$.  Then the postmeasurement states $\{\bbI_{A_1}\otimes M\ket{\psi}\mid\ket{\psi}\in \mathop{\cup}\limits_{1\leq i\leq 6}\cA_i\}$ should be mutually orthogonal. We can show that $E\propto \bbI$ from Table~\ref{Table:333}. This completes this proof.

\twocolumngrid	

\end{document}